\documentclass[11pt]{article}
\usepackage[utf8]{inputenc}
\usepackage[T1]{fontenc}
\usepackage{lmodern}
\usepackage[DIV=11]{typearea}
\usepackage{amssymb}
\usepackage{amsmath}
\usepackage{amsthm}
\usepackage{thmtools}
\usepackage{enumerate}
\usepackage{microtype}
\usepackage[basic]{complexity}
\usepackage[ruled,vlined,linesnumbered,nokwfunc,algo2e]{algorithm2e_}
\usepackage{xcolor}
\usepackage{xspace}
\usepackage{paralist}
\usepackage{booktabs}
\usepackage{multirow}
\usepackage{times}

\usepackage{todonotes}

\usepackage{tikz}
\usetikzlibrary{arrows,shapes,decorations,shadows}
\tikzstyle{player}=[draw, thick, circle, fill=gray!15,inner sep=2pt, minimum width=12pt]
\tikzstyle{vplayer}=[draw, thick, circle, fill=gray!15,inner sep=0.5pt, minimum width=12pt]
\tikzstyle{random}=[draw, thick, diamond, rounded corners, fill=gray!15,inner sep=2pt, minimum width=12pt]
\tikzstyle{vrandom}=[draw, thick, diamond, rounded corners, fill=gray!15,inner sep=0.5pt, minimum width=12pt]

\usepackage[
	style=alphabetic,
	backref=true,
	doi=false,
	url=false,
	maxcitenames=3,
	mincitenames=3,
	maxbibnames=10,
	minbibnames=10,
	backend=biber,
	sortlocale=en_US,
	sorting=nyt,
	bibencoding=ascii
]{biblatex}

\usepackage[ocgcolorlinks]{hyperref}

\makeatletter
\g@addto@macro\bfseries{\boldmath}
\makeatother

\makeatletter
\g@addto@macro\mdseries{\unboldmath}
\g@addto@macro\normalfont{\unboldmath}
\g@addto@macro\rmfamily{\unboldmath}
\g@addto@macro\upshape{\unboldmath}
\makeatother

\newcommand{\myhref}[1]{%
  \iffieldundef{doi}
    {\iffieldundef{url}
       {#1}
       {\href{\strfield{url}}{#1}}}
    {\href{http://dx.doi.org/\strfield{doi}}{#1}}%
}

\DeclareFieldFormat{title}{\myhref{\mkbibemph{#1}}}
\DeclareFieldFormat
  [article,inbook,incollection,inproceedings,patent,thesis,unpublished]
  {title}{\myhref{\mkbibquote{#1\isdot}}}

\makeatletter
\AtBeginDocument{%
    \newlength{\temp@x}%
    \newlength{\temp@y}%
    \newlength{\temp@w}%
    \newlength{\temp@h}%
    \def\my@coords#1#2#3#4{%
      \setlength{\temp@x}{#1}%
      \setlength{\temp@y}{#2}%
      \setlength{\temp@w}{#3}%
      \setlength{\temp@h}{#4}%
      \adjustlengths{}%
      \my@pdfliteral{\strip@pt\temp@x\space\strip@pt\temp@y\space\strip@pt\temp@w\space\strip@pt\temp@h\space re}}%
    \ifpdf
      \typeout{In PDF mode}%
      \def\my@pdfliteral#1{\pdfliteral page{#1}}
      \def\adjustlengths{}%
    \fi
    \ifxetex
      \def\my@pdfliteral #1{}
      \def\adjustlengths{\setlength{\temp@h}{-\temp@h}\addtolength{\temp@y}{1in}\addtolength{\temp@x}{-1in}}%
    \fi%
    \def\Hy@colorlink#1{%
      \begingroup
        \ifHy@ocgcolorlinks
          \def\Hy@ocgcolor{#1}%
          \my@pdfliteral{q}%
          \my@pdfliteral{7 Tr}
        \else
          \HyColor@UseColor#1%
        \fi
    }%
    \def\Hy@endcolorlink{%
      \ifHy@ocgcolorlinks%
        \my@pdfliteral{/OC/OCPrint BDC}%
        \my@coords{0pt}{0pt}{\pdfpagewidth}{\pdfpageheight}%
        \my@pdfliteral{F}
        \my@pdfliteral{EMC/OC/OCView BDC}%
        \begingroup%
          \expandafter\HyColor@UseColor\Hy@ocgcolor%
          \my@coords{0pt}{0pt}{\pdfpagewidth}{\pdfpageheight}%
          \my@pdfliteral{F}
        \endgroup%
        \my@pdfliteral{EMC}%
        \my@pdfliteral{0 Tr}
        \my@pdfliteral{Q}%
      \fi
      \endgroup
    }%
}
\makeatother

\colorlet{DarkRed}{red!50!black}
\colorlet{DarkGreen}{green!50!black}
\colorlet{DarkBlue}{blue!50!black}

\hypersetup{
	linkcolor = DarkRed,
	citecolor = DarkGreen,
	urlcolor = DarkBlue,
	bookmarksnumbered = true,
	linktocpage = true
}

\newcommand{\set}[1]{\{#1\}}
\newcommand{\lu}{\textup{(}}
\newcommand{\ru}{\textup{)}}
\newcommand{\upbr}[1]{\lu #1\ru}
\newcommand{\objsty}[2]{\textrm{#1}(#2)}
\newcommand{\Reach}[1]{\objsty{Reach}{#1}}
\newcommand{\Safe}[1]{\objsty{Safe}{#1}}
\newcommand{\Buchi}[1]{\objsty{Büchi}{#1}}
\newcommand{\coBuchi}[1]{\objsty{coBüchi}{#1}}
\newcommand{\target}{T\xspace}
\newcommand{\pre}{\mathsf{Pre}\xspace}
\newcommand{\post}{\mathsf{Post}\xspace}
\newcommand{\dist}{d}
\newcommand{\dia}{D}
\newcommand{\adia}{\widetilde{D}}
\newcommand{\distx}[2]{\overrightarrow{N}_{#1}(#2)}
\newcommand{\distxto}[2]{\overleftarrow{N}_{#1}(#2)}
\newcommand{\depthout}[1]{\overrightarrow{B}(#1)}
\newcommand{\depthin}[1]{\overleftarrow{B}(#1)}
\newcommand{\highdeg}[1]{V_{#1}}
\newcommand{\fw}[1]{FW(#1)}
\newcommand{\bw}[1]{BW(#1)}

\theoremstyle{plain}
\declaretheorem[numberwithin=section]{theorem}
\declaretheorem[numberlike=theorem]{lemma}
\declaretheorem[numberlike=theorem]{corollary}
\declaretheorem[numberlike=theorem]{Remark}
\newtheorem{proposition}[theorem]{Proposition}
\newtheorem{reduction}[theorem]{Reduction}

\newtheorem{observation}[theorem]{Observation}
\newtheorem{Definition}[theorem]{Definition}

\usepackage{authblk}
\title{Lower Bounds for Symbolic Computation on Graphs:\\
\Large Strongly Connected Components, Liveness, Safety, and Diameter}
\author[1]{Krishnendu Chatterjee}
\author[2]{Wolfgang Dvo{\v r}{\' a}k}
\author[3]{Monika Henzinger}
\author[3,4]{Veronika~Loitzenbauer}
\affil[1]{IST Austria}
\affil[2]{TU Wien, Institute of Information Systems, Vienna, Austria}
\affil[3]{University of Vienna, Faculty of Computer Science, Vienna, Austria}
\affil[4]{Bar-Ilan University}
\date{}

\hypersetup{
	pdftitle = {Lower Bounds for Symbolic Computation on Graphs},
	pdfauthor = {Krishnendu Chatterjee, Wolfgang Dvorak, Monika Henzinger, Veronika Loitzenbauer}
}
\newcommand{\citet}[1]{\citeauthor{#1}~\cite{#1}}

\begin{document}
\pagenumbering{gobble}
\maketitle
\begin{abstract}
A model of computation that is widely used in the formal analysis of reactive
systems is \emph{symbolic algorithms}. 
In this model the access to the input graph is restricted to consist of 
\emph{symbolic operations}, which are expensive in comparison to the standard RAM operations. 
We give lower bounds on the number of symbolic operations for basic graph
problems such as the computation of the strongly connected components and of the 
approximate diameter as well as for fundamental problems in model checking such 
as safety, liveness, and co-liveness. Our lower bounds are linear in the number 
of vertices of the graph, even for constant-diameter graphs. 
For none of these problems lower bounds on the number 
of symbolic operations were known before. The lower bounds show an interesting 
separation of these problems from the reachability problem, which can be solved 
with $O(\dia)$ symbolic operations, where $\dia$ is the diameter of the graph.
 
Additionally we present an approximation algorithm for the graph diameter which 
requires $\tilde{O}(n \sqrt{\dia})$ symbolic steps to achieve a $(1+\epsilon)$-approximation
for any constant $\epsilon > 0$. This compares to $O(n \cdot \dia)$ symbolic steps for the
(naive) exact algorithm and $O(\dia)$ symbolic steps for a 2-approximation. 
Finally we also give a refined analysis of the strongly connected components 
algorithms of~\cite{GentiliniPP08}, showing that it uses an optimal number of 
symbolic steps that is proportional to the sum of the diameters of the 
strongly connected components.
\end{abstract}

\clearpage
\pagenumbering{arabic}
\pagestyle{plain}
\setcounter{page}{1}

\section{Introduction}
Graph algorithms are central in the formal analysis of reactive systems. A reactive system consists of a set of variables and
a state of the system corresponds to a set of valuations, one for each of these variables. This naturally induces a directed graph:
Each vertex represents a state of the system and each directed edge represents a state transition that is possible in the system.
As the number of vertices is exponential in the number of variables of the system, these graphs are huge and, thus, they are usually not explicitly represented
during their analysis. Instead they are {\em implicitly represented} using e.g., binary-decision diagrams (BDDs) \cite{Bryant86,Bryant92}. To avoid considering specifics of the implicit representation and their manipulation,
an elegant theoretical model for algorithms that work on this implicit representation has been developed, called {\em symbolic algorithms} (see e.g.\ \cite{BurchCMDH90,ClarkeMCH96,Somenzi99,ClarkeGP99,ClarkeGJLV03,GentiliniPP08,ChatterjeeHJS13}).
In this paper we will give novel upper and (unconditional) lower bounds on the number of operations required by a symbolic algorithm for solving classic graph-algorithmic questions, such as computing the strongly connected components and the (approximate) diameter, as well as for graph-algorithmic questions that are important in the analysis of reactive systems, such as safety, liveness, and co-liveness objectives. Our lower bounds are based on new reductions of problems from
communication complexity to symbolic algorithms.

{\em Symbolic algorithms.}
A symbolic algorithm is allowed to use the same mathematical, logical, and memory access operations as a regular RAM algorithm, except for the
access to the input graph: It
 is not given  access to the input graph through an adjacency list or adjacency matrix representation but instead 
{\em only} through two types of {\em symbolic operations}:
\begin{compactenum}
\item {\em One-step operations Pre and Post:} Each {\em predecessor \upbr{Pre}} 
(resp., {\em successor \upbr{Post}}) operation is given a set $X$ of vertices and
returns the set of vertices~$Y$ with an edge to (resp., edge from) 
some vertex of~$X$.
\item {\em Basic set operations:} Each basic set operation is given one or two sets of vertices and performs a union, intersection, or complement on these sets.
\end{compactenum}
An initial set of vertices is given as part of the input, often consisting of a single vertex.

Symbolic operations are more expensive than the non-symbolic operations and 
thus one is mainly interested in the number of symbolic operations of such an algorithm 
(and the exact number of non-symbolic operations is often neglected).
Moreover, as the symbolic model is motivated by the compact representation of huge graphs,
we aim for symbolic algorithms that only store $O(1)$ or $O(\log n)$ many sets of vertices as 
otherwise algorithms become impractical due to the huge space requirements. 
Additionally, every computable graph-algorithmic question can be solved with $2n$ symbolic one-step operations when storing $O(n)$ many sets (and allowing an unbounded number of non-symbolic operations): 
For every vertex $v$ perform a {\em Pre} and a {\em Post} operation, store the results, which represent the full graph, and then compute the solution on this graph, using only non-symbolic operations.
Note, however, that our lower bounds do not depend on this requirement, i.e., they also apply to symbolic algorithms that store an arbitrary number of sets.
Furthermore the basic set operations (that only require vertices, i.e., the current state variables) are computationally much less 
expensive than the one-step operations (that involve both variables of the current and of the next state).
Thus, to simplify the analysis of symbolic
algorithms, we only analyze the number of one-step operations in the lower bounds that we present. For all upper bounds in prior work and in our work
the number of basic-set operations is at most linear in the number of one-step operations.

There is an interesting relationship between the two types of symbolic operations and Boolean matrix-vector operations: Interpreting the edge relationship as a Boolean matrix 
and a vertex set as a Boolean vector, the one-step operations correspond to (left- and right-sided)
matrix-vector multiplication, where the matrix is the adjacency matrix,
and basic set operations correspond to basic vector manipulations. Thus, an equivalent way of interpreting symbolic algorithms is by saying that the access to the graph is only allowed by performing a Boolean matrix-vector multiplication, where the vector represents a set of vertices and the matrix is the adjacency matrix.

Note also that there is a similarity to the CONGEST and the LOCAL model in synchronous 
distributed computation, as in these models each vertex in a synchronous network knows 
all its neighbors and can communicate with all of them in one round (exchanging
$O(\log n)$ bits in the CONGEST model), and the algorithmic complexity is measured in
rounds of communication. While in these models all neighbors of every individual 
vertex, i.e., all {\em edges} of the graph, can be determined in one round of 
communication, in the symbolic model this might require $n$ {\em Pre} and $n$ 
{\em Post} operations, each on an singleton set. 
Thus, determining (and storing) all edges of a symbolically represented graph is
expensive and we would ideally like to have algorithms that use sub-linear
(in the number of vertices) many symbolic one-step operations. 

{\em Objectives.} First we formally introduce the most relevant graph-algorithmic questions from the analysis of reactive systems \cite{MannaP92}. 
 Given a graph $G=(V,E)$ and a starting vertex $s \in V$, let ${\cal P}_s$ be the set of infinite paths in
 $G$ starting from $s$. Each objective corresponds to a set of requirements on an infinite path and the question that needs to be decided by the algorithm is whether there is a path in ${\cal P}_s$ that
 satisfies these requirements, in which case we say the path {\em satisfies the objective.}
 An objective $A$ is the {\em dual} of an objective $B$ if a path satisfies $A$ iff it does not satisfy $B$.

Let $\target \subseteq V$ be a set of target vertices given as input.
The most basic objective is \emph{reachability} where 
an infinite path satisfies the objective if the path visits a 
vertex of~$\target$ {\em at least once}.
The dual \emph{safety} objective 
is satisfied by infinite paths that only visit vertices of~$\target$.
The next interesting objective is the \emph{liveness} (aka \emph{B\"uchi}) objective 
that 
requires an infinite path to visit some vertex of~$\target$ \emph{infinitely often}.
The dual \emph{co-liveness} (aka \emph{co-B{\"u}chi}) objective 
requires an infinite path to eventually only visit vertices in $\target$. 
Verifying these objectives are the most fundamental graph-algorithmic questions 
in the analysis of reactive systems.

Computing the strongly connected components (SCCs) is at the heart of the fastest algorithms for liveness and co-liveness: For example,
there is a reduction from liveness to the computation of SCCs 
that takes symbolic steps in the order of the diameter of the graph. 
Thus, determining the symbolic complexity of SCCs also settles the symbolic complexity of liveness.

Furthermore, the diameter computation plays a crucial role in applications such as bounded
model-checking~\cite{BiereCCSZ03}, where the goal is to analyze the system for a bounded number
of steps, and it suffices to choose the diameter of the graph as bound.
Second, in many scenarios, such as in hardware verification, the graphs have 
small diameter, and hence algorithms that can detect if this is the case and then exploit the small diameter are relevant~\cite{BiereCCSZ03}. 
Motivated by these applications, we define the diameter of a graph as the largest 
finite distance in the graph, which coincides with the usual graph-theoretic 
definition on strongly connected graphs and is more general otherwise.

Note that linear lower bounds for the number of symbolic operations 
are non-trivial, since a one-step operation can involve {\em all} edges. 
For example, to determine all the neighbors of a given vertex $v$ takes one symbolic operation, while it takes $O(\deg(v))$ many operations in the classic setting. 
In the following we use $n$ to denote the number of vertices of a graph~$G$ 
and $\dia(G) = \dia$ to denote its diameter.

\smallskip\noindent{\em Previous results.}
To the best of our knowledge, no previous work has established lower bounds 
for symbolic computation.  

There is some prior work on establishing upper bounds on the number of symbolic operations:
In~\cite{GentiliniPP08} a symbolic algorithm that computes the SCCs
with $O(n)$ symbolic operations is presented.
This algorithm leads to an algorithm for liveness and co-liveness with $O(n)$
symbolic operations and improves on earlier work by~\cite{BloemGS06}, which 
requires $O(n \log n)$ symbolic operations.

Note that for the reachability  objective the straightforward algorithm requires $O(\dia)$
symbolic operations: Starting from the set containing only the start vertex $s$,
repeatedly perform a {\em Post}-operation until $\target$ is reached.
For safety the straightforward algorithm takes $O(n)$ symbolic operations:
Iteratively remove from $F$ vertices that do not have an outgoing
edge to another vertex of $F$, i.e., vertices of $F \setminus \pre(F)$, until a fixed point is reached.

Finally, there is a trivial algorithm for computing the diameter of the graph: 
Simply determine the depth of a breadth-first search from every vertex and output 
the maximum over all depths.
Computing the depth of a breadth-first search can be done with $O(\dia)$ 
many symbolic steps, thus this requires $O(n \dia)$ many symbolic steps in total.
In a strongly connected graph a 2-approximation of the diameter of the graph can 
be obtained by computing one breadth-first search from and one to
some arbitrary vertex and output the sum of the depths. 
This takes $O(\dia)$ symbolic steps.

\smallskip\noindent{\em Our contributions.} 
Our main contributions are novel lower bounds for the number of symbolic operations for many
of the above graph-algorithmic questions, leading to an interesting separation between seemingly similar problems.
\begin{compactenum}
\item For reachability objectives, the basic symbolic algorithm requires
$O(\dia)$ symbolic operations.
Quite surprisingly, we show that such diameter-based upper bounds are 
{\em not possible} for its dual problem,
namely safety, and are also not possible for liveness and co-liveness objectives. 
Specifically, we present tight lower bounds to show that, even for constant-diameter graphs, 
$\Omega(n)$ one-step symbolic operations are required for safety, liveness, 
and co-liveness objectives. 
In addition we establish tight bounds for symbolic operations required 
for the computation of SCCs, showing
a lower bound of $\Omega(n)$ for constant-diameter graphs.
See Table~\ref{tab:results_objectives_scc} for a summary of these results.

\begin{table}[t]
  \caption{Bounds on the number of required symbolic operations for different tasks. $\Theta(n)$ bounds hold even for graphs with constant diameter $\dia$.} \medskip
  \label{tab:results_objectives_scc}
  \centering
  \begin{tabular}{ c  c  c  c c}
    \toprule
    $\Reach{\target}$ & $SCC$ & $\Safe{\target}$ & $\Buchi{\target}$ & $\coBuchi{\target}$\\
    \midrule
    $\Theta(\dia)$ & $\Theta(n)$ & $\Theta(n)$ & $\Theta(n)$ & $\Theta(n)$ \\
    \bottomrule
  \end{tabular}
      \vspace{-5mm}
\end{table}

\item We show that even for strongly-connected constant-diameter graphs
approximating the diameter requires $\Omega(n)$ symbolic steps.
More precisely, 
any $(3/2-\varepsilon)$-approximation algorithm requires $\Omega(n)$ 
symbolic one-step operations, even 
for undirected and connected graphs 
with constant diameter.
We also give a novel upper bound: We present a $(1+\varepsilon)$-approximation
algorithm for any constant $\varepsilon > 0$ that
takes $\widetilde O(n \sqrt{\dia})$ symbolic steps.
This can be compared to the trivial $O(\dia)$ 2-approximation algorithm
and the $O(n \dia)$ exact algorithm. 
Notice that for explicitly represented graphs the approximation of the diameter
is already hard for constant-diameter graphs while in the symbolic model there
exists a trivial $O(n)$ upper bound in this case, thus showing a lower bound of $\Omega(n)$ 
is non-trivial.
See Table~\ref{tab:results_diameter} for a summary of these results.

\begin{table}[t]~\label{tab:results_diameter}
  \caption{Bounds on the number of symbolic operations for approximating the diameter of a graph. The lower bounds even hold for strongly connected graphs with constant diameter $\dia$.}\medskip
  \centering
  \begin{tabular}{l| c c c c}
    \toprule
    approx. & exact & $1 + \varepsilon$ & $3/2-\varepsilon$ & $2$\\
    \midrule
    upper bound & $O(n\cdot \dia)$ & $\widetilde O(n \sqrt{\dia})$ & $\widetilde O(n \sqrt{\dia})$ & $O(\dia)$\\
    lower bound & $\Omega(n)$ & $\Omega(n)$ & $\Omega(n)$\\
    \bottomrule
  \end{tabular}
  \vspace{-5mm}
\end{table}

\item Finally we give a refined analysis of the number of symbolic steps required 
for computing strongly connected components based on a different problem parameter. 
Let $SCCs(G)$ be the set of all SCCs of $G$ and $\dia_C$ the diameter of the 
strongly connected component $C$. We give matching upper and
lower bounds showing that the SCCs can be computed with
$\Theta(\sum_{C \in SCCs(G)} (\dia_C + 1))$ symbolic steps. Note that 
$\sum_{C \in SCCs(G)} (\dia_C + 1)$ can be a factor $n$ larger than $\dia(G)$.
\end{compactenum}

\smallskip\noindent{\em Key technical contribution.}
Our key technical contribution is based on the novel insight that lower bounds
for communication complexity can be used to establish lower bounds for symbolic
computation. We feel that this connection is of interest by itself and might lead 
to further lower bounds for symbolic algorithms.

Our lower bounds are by two kinds of reductions, both from the communication complexity problem of Set Disjointness with $k$ elements.
First, we give reductions that construct graphs such that one-step operations can be computed 
with $O(1)$ bits of communication between Alice and Bob and thus allow for linear lower bounds
on the number symbolic operations.
Second, we give a reduction that constructs a graph with only $\sqrt{k}$ many vertices, i.e., $n=\sqrt{k}$, but
allows one-step operations to require $O(n)$ bits of communication. 
This again results in linear lower bounds on the number of symbolic operations.

\section{Preliminaries}

\noindent{\bf Symbolic Computation.}
We consider symbolic computation on graphs. Given an input graph $G = (V, E)$ 
and a set of vertices $S \subseteq V$, 
the graph~$G$ can be accessed only by the following two types of operations:
\begin{compactenum}
	\item \emph{Basic set operations} like $\cup$, $\cap$, $\setminus$, $\subseteq$, and $=$;
	\item \emph{One-step operations} to obtain the predecessors or successors of the 
	vertices of~$S$ in~$G$. In particular we define the operations
	\begin{align*}
		\pre(S) =\{v \in V \mid \exists s \in S: (v, s) \in E\} \ \
		\text{and}  \ \
		\post(S)=\{v \in V \mid \exists s \in S: (s, v) \in E\}\,.
	\end{align*}
\end{compactenum}
In the applications the basic set operations are much cheaper as compared to the 
one-step operations. Thus we aim for lower bounds on the number of one-step
operations, while not accounting for set operations. In all our upper bounds the 
number of set operations is at most of the same order as the number of one-step operations.
Note that there is a one-to-one correspondence between a one-step operation and a Boolean
matrix-vector multiplication with the adjacency matrix 
and that for undirected graphs $\pre$ and $\post$ are equivalent.

\smallskip\noindent{\bf Communication Complexity Lower Bound for Set Disjointness.}
Our lower bounds are based on the known lower bounds for the communication
complexity of the \emph{Set Disjointness} problem. The classical symmetric 
two-party communication complexity model is as follows~\cite{KushilevitzN97}.
There are three finite sets $X,Y,Z$, the former two are possible inputs for a function $f: 
X \times Y \rightarrow Z$,
where the actual input $x\in X$ is only known by Alice, and the actual input $y \in Y$
is only known by Bob.
Alice and Bob want to evaluate a function 
$f(x,y)$ while sending as few bits as possible to each other.
The communication happens according to a fixed protocol, known to both players beforehand,
that determines which player sends which bits when, and when to stop. 

\smallskip\noindent\emph{Set Disjointness.}
In the Set Disjointness problem we have a universe $U = \set{0, 1, \ldots, k-1}$ of $k$ elements and
both sets $X$, $Y$ contain all bit vectors of length $k$, i.e., they represent all possible 
subsets of $U$ and are of size $2^k$.
Alice has a vector $x \in X$ and Bob has a vector $y \in Y$, and the function $f$ 
is defined as $f(x,y) = 1$ if for all $0 \leq i \leq k-1$ either $x_i=0$ or $y_i = 0$, and $f(x,y) = 0$ otherwise. 
We will sometimes use $S_x$ and $S_y$ to denote the sets corresponding to the vectors $x$ and $y$,
i.e., $S_x=\{i \mid x_i=1\}$ and $S_y=\{i \mid y_i=1\}$ and $f(x,y) = 1$ iff $S_x \cap S_y = \emptyset$.
We next state a fundamental lower bound for the communication complexity of the Set Disjointness problem 
which will serve as basis for our lower bounds on the number of symbolic operations.

\begin{theorem}[\cite{KalyanasundaramS92,Razborov92,Bar-YossefJKS04,HastadW07,KushilevitzN97}]\label{th:lbdisj}
	Any \upbr{probabilistic bounded error or deterministic} protocol for the Set 
	Disjointness problem sends $\Omega(k)$ bits in the worst case over all inputs.
\end{theorem}

\section{Lower Bounds}
In this section we present our lower bounds, which are the main results of the paper.

\subsection{Lower Bounds for Computing Strongly Connected Components}\label{sec:lb_SCCs}
We first consider the problem of computing the strongly connected components (SCCs) of a 
symbolically represented graph. 
The best known symbolic algorithm is  by Gentilini et al.~\cite{GentiliniPP08} and computes the SCCs of a Graph G
with $O(\min(n, \dia \cdot |SCCs(G)|))$ symbolic one-step operations and thus matches the linear running time 
of the famous Tarjan algorithm~\cite{Tarjan72} in the non-symbolic world.

We provide lower bounds showing that the algorithm is essentially optimal, 
in particular we show that $O(\dia)$ algorithms are impossible.
These lower bounds are by reductions from the communication complexity problem of
Set Disjointness to computing 
SCCs in a specific graph. In particular, we show that any algorithm that computes SCCs with $o(n)$
symbolic one-step operations would imply a communication protocol for the Set Disjointness problem with $o(k)$
communication.

\begin{reduction}\label{red:scc2}
  Let $(x,y)$ be  an instance of Set Disjointness and let w.l.o.g.\ $k= \ell \cdot \bar{k}$
  for some integers $\ell, \bar{k}$. We construct a directed graph $G = (V, E)$  with 
  $n = k + \ell$ vertices and $O(n^2)$ edges as follows.
    \upbr{1} The vertices are given by $V= \bigcup_{i=0}^{\ell-1} V_i$
    with $V_i=\{v_{i,0},\dots,v_{i,\bar{k}}\}$.
    \upbr{2} There is an edge from $v_{i,j}$ to $v_{i',j'}$ if either $i<i'$ or $i=i'$ and $j < j'$.
    \upbr{3} For $0 \leq i < \ell$, $0 \leq j < \bar{k}$ there is an edge from 
	     $v_{i,j+1}$ to $v_{i,j}$ iff $x_{i\cdot \bar{k} + j}=0$ or $y_{i\cdot \bar{k} + j}=0$. 
\end{reduction}

In our communication protocol both Alice and Bob compute the number of SCCs on 
the graph from Reduction~\ref{red:scc2}, according to a given algorithm. 
While both know all the vertices of the graph, they do not know all the edges 
(some depend on both $x$ and $y$) and 
thus whenever such an edge is relevant for the algorithm,
Alice and Bob have to communicate with each other.
We show that the graph is constructed such that for each subset $S\subseteq V$
the operations $\pre(S)$ and $\post(S)$
can be computed with only four bits of communication between Alice and Bob.

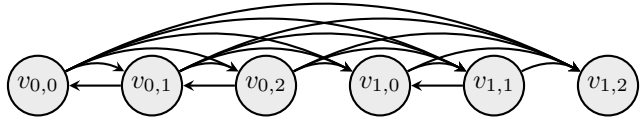
\begin{figure}[bt]
 \centering
\begin{tikzpicture}[>=stealth, yscale=0.9]
  \small
  		\path 	(0,0)node[player](v00){$v_{0,0}$}
			++(1.5,0)node[player](v01){$v_{0,1}$}
			++(1.5,0)node[player](v02){$v_{0,2}$}
			++(1.5,0)node[player](v10){$v_{1,0}$}
			++(1.5,0)node[player](v11){$v_{1,1}$}
			++(1.5,0)node[player](v12){$v_{1,2}$}
			;
		\path [->, bend left, thick]
			(v00) edge (v01)
			(v00) edge (v02)
			(v00) edge (v10)
			(v00) edge (v11)
			(v00) edge (v12)
			(v01) edge (v02)
			(v01) edge (v10)
			(v01) edge (v11)
			(v01) edge (v12)
			(v02) edge (v10)
			(v02) edge (v11)
			(v02) edge (v12)
			(v10) edge (v11)
			(v10) edge (v12)
			(v11) edge (v12)
			;
		\path [->, thick]
			(v11) edge (v10)
			(v02) edge (v01)
			(v01) edge (v00)
			;
  \end{tikzpicture}
  \caption{Reduction~\ref{red:scc2} for $k=4, \ell=2, S_x=\{2, 3\}, S_y=\{0, 1, 3\}$}
\end{figure}

\begin{theorem}\label{thm:scc2}
  Any \upbr{probabilistic bounded error or deterministic} symbolic algorithm
  that computes the SCCs of graphs with $n$ vertices
  needs $\Omega(n)$ symbolic one-step operations.
  Moreover, for a graph with the set $SCCs(G)$ of SCCs and diameter $\dia$
  any algorithm needs $\Omega(|SCCs(G)|\cdot \dia)$ symbolic one-step operations.
\end{theorem}

We first show that Reduction~\ref{red:scc2} is a valid reduction from the Set Disjointness problem to an SCC problem. The missing proofs are given in Section~\ref{sec:lb_SCCs_proofs}.

\begin{lemma}\label{lem:scc1}
  $f(x,y)=1$ iff the graph constructed in Reduction~\ref{red:scc2} has exactly $\ell$ SCCs.
\end{lemma}

The critical observation for the proof of Theorem~\ref{thm:scc2} is that any
algorithm that computes SCCs with $N$ many symbolic one-step operations 
implies the existence of a communication protocol for Set Disjointness 
that only requires $O(N)$ communication.

\begin{lemma}\label{lem:scc2}
 For any algorithm that computes SCCs with $N$ symbolic one-step operations 
 there is a communication protocol for Set Disjointness that requires $O(N)$ communication.
\end{lemma}
\begin{proof}
 In our communication protocol both Alice and Bob consider the graph from Reduction~\ref{red:scc2}.
 We call edges of the graph that are present independently of $x$ and $y$ 
 \emph{definite} edges and 
 edges whose presence depends on $x$ and $y$ \emph{possible} edges.
 
 Both Alice and Bob  execute the given symbolic algorithm to decide whether the graph has $\ell$ SCCs (cf.~Lemma~\ref{lem:scc1}). 
 As both know all the vertices, they can execute set operations without communicating.
 Communication is only needed when executing symbolic one-step operations, 
 since for these some of the possible edges might affect the outcome.
 
 We next argue that each symbolic one-step operations can be executed with a constant number of bits of communication.
 First notice that as both Alice and Bob execute the same algorithm simultaneously,
 they both know the input set to an operation and 
 they only need communication about the possible edges that can change the output.
 Both can independently identify these possible edges and 
 they can decide whether such an edge exists by sending one bit each.
 We next argue that for each one-step operation we need to consider at most two possible edges.
 For this we consider the vertices $v_{i,j}$ in their linear ordering given 
 by $i \cdot (\bar{k}+1) + j$, e.g., $v_{0,0}=v_0$ and $v_{\ell-1,\bar{k}} = v_{k+\ell-1}$.
 \smallskip
 
	$\post$ operation: Let $S$ be the input set and let $v_{\min}$ the vertex 
	with the \emph{minimum} index in~$S$.
	Then we have $\{v_{\min+1}, \dots v_{k+\ell-1}\} \subseteq \post(S)$ and 
	potentially also $v_{\min}$ and $v_{\min-1}$
	can be in $\post(S)$, but no other vertices.
	That is, we have $\{v_{\min+1}, \dots v_{k + \ell-1}\} \subseteq \post(S) \subseteq \{v_{\min-1}, \dots v_{k + \ell-1}\}$.
	To decide whether $v_{\min}$ is in $\post(S)$, we first check whether $v_{\min+1} \in S$
	and if so we check whether the edge $(v_{\min+1},v_{\min})$ is present.
	To decide $v_{\min-1} \in \post(S)$, we check whether the edge $(v_{\min},v_{\min-1})$ is present.
	That is, to compute $\post(S)$ we only access two possible edges.\smallskip
    
	$\pre$ operation: Let $S$ be the input set and let $v_{\max}$ the vertex with the \emph{maximum} index in~$S$.
	Then we have $\{v_{0}, \dots v_{\max-1}\} \subseteq \pre(S)$ and potentially 
	also $v_{\max}$ and $v_{\max+1}$
	can be in $\pre(S)$, but no other vertices.
	That is, we have $\{v_{0}, \dots v_{\max-1}\} \subseteq \pre(S) \subseteq \{v_{0}, \dots v_{\max+1}\}$.	
	To decide whether $v_{\max}$ is in $\pre(S)$, we first check whether $v_{\max-1} \in S$
	and if so we check whether the edge $(v_{\max},v_{\max-1})$ is present.
	To decide if $v_{\max+1} \in \pre(S)$, we check whether the edge $(v_{\max+1},v_{\max})$ is present.
	That is, we can compute $\pre(S)$ with accessing only two possible edges.\smallskip

 By the above we have that a symbolic algorithm with $N$ one-step operations
 gives rise to a communication protocol for 
 Set Disjointness with $O(N)$ bits of communication.
\end{proof}

  By Lemma~\ref{lem:scc2} we have that any algorithm computing SCCs with $o(n)$
   symbolic one-step operations would contradict Theorem~\ref{th:lbdisj}.
  Now inspecting the graph of Reduction~\ref{red:scc2}, we observe that its diameter 
  is equal to $\bar{k}$, which leads to the following lower bounds.
  For $\ell=k/2$ the graph has diameter $2$ and thus the $\Omega(n)$ holds even 
  for graphs of constant diameter.
  On the other side, for $\ell=1$ disjoint sets $S_x$ and $S_y$ 
  correspond to strongly connected graphs 
  and thus the $\Omega(n)$ lower bounds also holds for graphs with a bounded number of SCCs,
  i.e., there are no $O(|SCCs(G)|)$ symbolic algorithms.   
  Finally for $\ell=\sqrt{k}$ we obtain a $\Omega(|SCCs(G)|\cdot \dia)$ lower bound.

\begin{Remark}
   By the above no algorithm can compute SCCs with $f(\dia)\cdot n^{o(1)}$ or 
   $f(|SCCs(G)|)\cdot n^{o(1)}$ 
   symbolic one-step operations for any function $f$.
   In contrast, if we consider both parameters simultaneously, there is an $O(|SCCs(G)| \cdot \dia)$ symbolic algorithm.
\end{Remark}

The above lower bounds for computing SCCs match the $O(\min(n, \dia \cdot |SCCs(G)|))$ bound by the algorithm of
Gentilini et al.~\cite{GentiliniPP08}.
One way to further improve the algorithm is to not consider the diameter of the 
whole graph but the diameter $\dia_C$ of each single SCC~$C$. 
In that direction the previous reduction already gives us an $\Omega(\sum_{C \in SCCs(G)} (\dia_C))$ lower bound
and we will next improve it to an $\Omega(\sum_{C \in SCCs(G)} (\dia_C+1))$ lower bound (i.e., it is $\Omega(n)$ even if $\sum_{C \in SCCs(G)} (\dia_C) \in O(1)$).
These two bounds differ if the graph has a large number of trivial SCCs.
Thus we next give a reduction that constructs a graph that has only trivial SCCs
if $S_x$ and $S_y$ are disjoint.

\begin{reduction}\label{red:scc3}
  Given an instance $(x,y)$ of Set Disjointness,
  we construct a directed graph $G = (V, E)$ with $n = k+1$ vertices and $O(n^2)$ edges as follows.
    \upbr{1} The vertices are given by $V=\{v_0, v_1, \dots ,v_k\}$.
    \upbr{2} There is an edge from $v_i$ to $v_j$ for $i<j$.
    \upbr{3} For $0 \leq j \leq k-1$ there is an edge from $v_{j+1}$ to $v_j$ iff $x_j=1$ and $y_j=1$. 
\end{reduction}

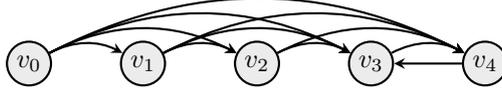
\begin{figure}[t]
 \centering
  \begin{tikzpicture}[>=stealth, yscale=0.9]
  \small
  		\path 	(0,0)node[player](v0){$v_0$}
			++(1.5,0)node[player](v1){$v_1$}
			++(1.5,0)node[player](v2){$v_2$}
			++(1.5,0)node[player](v3){$v_3$}
			++(1.5,0)node[player](v4){$v_4$}
			;
		\path [->, bend left, thick]
			(v0) edge (v1)
			(v0) edge (v2)
			(v0) edge (v3)
			(v0) edge (v4)
			(v1) edge (v2)
			(v1) edge (v3)
			(v1) edge (v4)
			(v2) edge (v3)
			(v2) edge (v4)
			(v3) edge (v4)
			;
		\path [->, thick]
			(v4) edge (v3)
			;
  \end{tikzpicture}
  \caption{Reduction~\ref{red:scc3} for $k=4, \ell=2, S_x = \{2, 3\}, S_y = \{0, 1, 3\}$}
  \label{fig:scc3}
\end{figure}

\begin{theorem}\label{thm:scc3}
  Any \upbr{probabilistic bounded error or deterministic} symbolic algorithm
  that computes the SCCs needs $\Omega(|SCCs(G)| + \sum_{C \in SCCs(G)} \dia_C)$ symbolic one-step operations.
\end{theorem}

\subsection{Lower Bounds for Liveness, Reachability, and Safety Objectives}\label{sec:lb_objectives}
In this section we extend our lower bounds for SCC computation to Liveness, Reachability, and Safety Objectives on graphs.\smallskip

\noindent\emph{Lower Bounds for Reachability.} 
The lower bounds for Reachability are an immediate consequence from our lower bounds for SCC computation
in Theorem~\ref{thm:scc2}. When setting $\ell=1$ in Reduction~\ref{red:scc2} then the vertex $v_{0,0}$ is reachable from all vertices iff
the graph is strongly connected iff the sets $S_x$ and $S_y$ are disjoint.

\begin{theorem}\label{thm:reach}
  Any \upbr{probabilistic bounded error or deterministic} symbolic algorithm
  that solves Reachability in graphs with diameter $\dia$ requires $\Omega(\dia)$ symbolic one-step operations.
\end{theorem}

\smallskip\noindent\emph{Lower Bounds for Liveness.} 
To show an $\Omega(n)$ lower bound for Liveness objectives which holds even 
for graphs of bounded diameter, we introduce another reduction. 
This reduction is again from the Set Disjointness Problem and also constructs
a graph such that $\pre$ and $\post$ operations can be executed with a constant number of bits of communication 
between Alice and Bob.

\begin{reduction}\label{red:buchi}
  Given an instance $(x,y)$ of Set Disjointness,
  we construct a directed graph $G = (V, E)$   with $n = k+1$ vertices and $O(n^2)$ edges as follows.
    \upbr{1} The vertices are given by $V=\{v_0, v_1, \dots ,v_k\}$.
    \upbr{2} There is an edge from $v_i$ to $v_j$ for $i<j$ and there is a loop edge $(v_k,v_k)$.
    \upbr{3} For $0 \leq j \leq k-1$ there is a loop edge $(v_j, v_j)$ iff $x_j=1$ and $y_j=1$. 
\end{reduction}

Notice that the graph in Reduction~\ref{red:buchi} has diameter $\dia=1$ 
and thus allows to show the lower bounds stated in Theorem~\ref{thm:buchi}
when considering $\target=\{v_0, v_1, \dots, v_{k-1}\}$, 
with the exception of \upbr{2} which is by Reduction~\ref{red:scc2} and $T=\{v_0\}$.

\begin{figure}[bth]
 \centering
  \begin{tikzpicture}[>=stealth]
  \small
  		\path 	(0,0)node[player](v0){$v_0$}
			++(1.5,0)node[player](v1){$v_1$}
			++(1.5,0)node[player](v2){$v_2$}
			++(1.5,0)node[player](v3){$v_3$}
			++(1.5,0)node[player](v4){$v_4$}
			;
		\path [->, bend left, thick]
			(v0) edge (v1)
			(v0) edge (v2)
			(v0) edge (v3)
			(v0) edge (v4)
			(v1) edge (v2)
			(v1) edge (v3)
			(v1) edge (v4)
			(v2) edge (v3)
			(v2) edge (v4)
			(v3) edge (v4)
			;
		\path [->, loop, out=-45, in=-135, looseness=4, thick]
			(v4) edge (v4)
			(v3) edge (v3)
			;
  \end{tikzpicture}
  \caption{Reduction~\ref{red:buchi} for $k=4, \ell=2, S_x=\{2, 3\}, S_y=\{0, 1, 3\}$}
  \label{fig:buchi}
\end{figure}
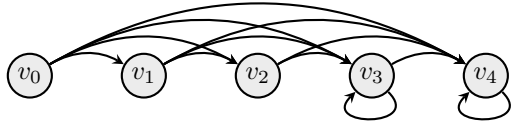

\begin{theorem}\label{thm:buchi}
  For any \upbr{probabilistic bounded error or deterministic} symbolic algorithm
  that solves $\Buchi{\target}$
  the following lower bounds on the 
  required number of symbolic one-step operations hold:
\upbr{1}~$\Omega(n)$ even for instances with constant $\dia$;
\upbr{2}~$\Omega(\dia)$ even for instances with $|\target|=1$;
\upbr{3}~$\Omega(|\target|)$ even for instances with constant $\dia$;
\upbr{4}~$\Omega(|\target| + \dia)$; and
\upbr{5}~$\Omega(|SCCs(G)| + \sum_{C \in SCCs(G)} \dia_C)$.
\end{theorem}

\smallskip\noindent\emph{Lower Bounds for co-Liveness and Safety.} 
The following lower bounds are by Reduction~\ref{red:buchi} (and variations of it) and  
the set of safe vertices $\target=\{v_0, v_1, \dots, v_{k-1}\}$.

\begin{theorem}\label{thm:coBuchi}
  For any \upbr{probabilistic bounded error or deterministic} symbolic algorithm
  that solves $\Safe{T}$ or $\coBuchi{T}$ 
  the following lower bounds on the required
  number of symbolic one-step operations hold:  
\upbr{1}~$\Omega(n)$ even for constant diameter graphs;
\upbr{2}~$\Omega(|\target|)$ even for constant diameter graphs; and
\upbr{3}~$\omega(\sum_{C \in SCCs(G)} (\dia_C + 1))$ even for constant diameter graphs.
\end{theorem}

Notice that the parameters diameter, number of SCCs, or diameters of SCCs do not help in the case of Safety. 
This is because every graph can be reduced to a strongly connected graph with diameter $2$
without changing the winning set as follows:
Add a new vertex $v$ that has an edge to and an edge from
all original vertices but do not add $v$ to the safe vertices~$T$.

We complete this section with a $\Omega(\dia)$ lower bound for $\coBuchi{T}$ which is by a variant of Reduction~\ref{red:scc2}.

\begin{proposition}\label{prop:coBuchi}
  Any \upbr{probabilistic bounded error or deterministic} symbolic algorithm
  that solves $\coBuchi{T}$ on graphs with diameter $\dia$ needs $\Omega(\dia)$
  symbolic one-step operations.
\end{proposition}

\subsection{Lower Bound for Approximate Diameter}\label{sec:lb_diameter}

\noindent\emph{The Approximate Diameter Problem.}
Let $G = (V, E)$ be a directed graph with $n$ vertices~$V$ and $m$ edges~$E$. 
Let $\dist(u, v)$ denote the shortest distance from $u \in V$ to $v \in V$ in~$G$,
i.e., the smallest number of edges of any path from $u$ to $v$ in~$G$.
Recall that we define the \emph{diameter} of~$G$ as 
the maximum of $\dist(u, v)$ over all pairs $u, v$ for which $u$ can reach~$v$ in~$G$.
\footnote{Usually the diameter is defined over all pairs
of vertices, not just the reachable ones, and is therefore $\infty$ if $G$ is not 
strongly connected. Our definition is more general since determining whether the
graph is strongly connected takes only $O(\dia)$ symbolic steps and additionally 
our definition is more natural in the symbolic setting as it 
provides an upper bound on the number of one-step operations needed until a fixed 
point is reached, which is an essential primitive in symbolic graph algorithms.}
We consider the problem of approximating the diameter~$\dia$ of a graph by a 
factor~$c$, where the goal is to compute an estimate $\adia$ such that 
$\dia / c \le \adia \le \dia$.
As undirected graphs are special cases of directed
graphs, the lower bound is presented for undirected graphs
and the upper bound for directed graphs (see Section~\ref{sec:upper}), 
i.e., both hold for undirected and directed graphs.

\smallskip\noindent\emph{Result.}
We show a lower bound of $\Omega(n)$ on the number of symbolic steps
needed to distinguish between a diameter of 2 and a diameter of 3, even in 
an undirected connected graph. The basic symbolic algorithm for computing 
the diameter exactly takes $O(n \cdot \dia)$ many symbolic steps. Thus our lower
bound is tight for constant-diameter graphs. 

\smallskip\noindent\emph{Outline Lower Bound.}
We show how to encode an instance of the Set Disjointness Problem
with a universe of size $k$ in an (undirected, connected) graph with
$\Theta(\sqrt{k})$ vertices and $\Theta(k)$ edges such that 1) in a communication protocol 
any symbolic one-step operation can be simulated with $\Theta(\sqrt{k})$ bits
and 2) the graph has diameter~2 if the two sets are disjoint and diameter~3 otherwise.
Thus the communication complexity lower bound of $\Omega(k)$ for Set Disjointness
implies a lower bound of $\Omega(\sqrt{k}) = \Omega(n)$ for the number
of symbolic one-step operations to compute a $(3/2 - \varepsilon)$-approximation
of the diameter of a graph with $n$ vertices.

\begin{reduction}\label{red:diasparse}
Let $(x,y)$ be an instance of the 
Set Disjointness problem of size $k$ and let $s = \sqrt{k}$. 
We construct an undirected graph $G = (V, E)$
  with $n = 3s + 2$ vertices and $O(k)$ edges as follows.
      \upbr{1} There are three sets $A, B, C$ with $s$ vertices each and two auxiliary 
      vertices~$u$ and $t$.
      We denote the $i$-th vertex of each of $A, B, C$ with a lowercase letter indicating the set and subscript~$i$.
      \upbr{2} There is an edge between $u$ and $t$ and between $u$
      and each vertex of $A$ and $B$ and between $t$ and each vertex of $C$.
      \upbr{3} For each $0 \le i < s$ there is an edge between $a_i \in A$ and $b_i \in B$.
      \upbr{4} For $0 \le \ell < k$ let $i, j < s$ be such that $\ell = i \cdot s + j$.
      There is an edge between $a_i \in A$ and $c_j \in C$ iff $x_\ell = 0$
      and there is an edge between $b_i \in B$ and $c_j \in C$ iff $y_\ell = 0 $.
\end{reduction}
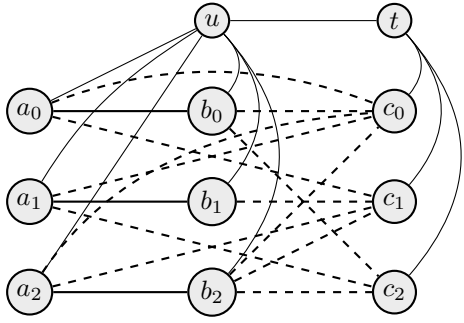
\begin{figure}
 \centering
 \begin{tikzpicture}[>=stealth,yscale=0.8,xscale=1.2]
  \small
		\draw  (2, 1.5)node[player](u){$u$};
		\draw  (4, 1.5)node[player](t){$t$};
  		\path 	(0,0)node[player](a0){$a_0$}
			++(0,-1.5)node[player](a1){$a_1$}
			++(0,-1.5)node[player](a2){$a_2$}
			;
  		\path 	(2,0)node[player](b0){$b_0$}
			++(0,-1.5)node[player](b1){$b_1$}
			++(0,-1.5)node[player](b2){$b_2$}
			;
  		\path 	(4,0)node[player](c0){$c_0$}
			++(0,-1.5)node[player](c1){$c_1$}
			++(0,-1.5)node[player](c2){$c_2$}
			;
		\path [-]
			(u) edge (t)
			(u) edge (a0)
			(u) edge[bend right = 12] (a1)
			(u) edge (a2)
			[bend left]
			(t) edge (c0)
			(t) edge (c1)
			(t) edge (c2)
			(u) edge (b0)
			(u) edge (b1)
			(u) edge (b2)
			;
		\path [-, thick]
			(a0) edge (b0)
			(a1) edge (b1)
			(a2) edge (b2)
			;
		\path [-, dashed, thick] 			
			(a0) edge[bend left] (c0)
			(a0) edge (c1)
			(a1) edge (c0)
			(a1) edge (c2)
			(a2) edge[bend left] (c0)
			(a2) edge (c1)
			;
		\path [-, dashed, thick] 			
			(b0) edge (c0)
			(b0) edge (c2)
			(b1) edge (c1)
			(b2) edge (c0)
			(b2) edge (c1)
			(b2) edge (c2)
			;
  \end{tikzpicture}
  \caption{Reduction~\ref{red:diasparse} for $k=9, s=3, S_x=\{2,4, 8\}, S_y=\{1,3,5\}$.}
\end{figure}

We first show that this graph has diameter~$2$ if $S_x$ and $S_y$ are 
disjoint and diameter~$3$ otherwise and then show how Alice can obtain a 
communication protocol for the Set Disjointness problem from any symbolic 
algorithm that can distinguish these two cases.
\begin{lemma}\label{lem:diasparse}
	Let $G = (V, E)$ be the graph given by Reduction~\ref{red:diasparse} and let 
	$D$ denote its diameter. If $S_x \cap S_y = \emptyset$, then $D = 2$, otherwise $D = 3$.
\end{lemma}

In the graph $G = (V, E)$ given by Reduction~\ref{red:diasparse} 
Alice knows all the vertices of the graph and all 
the edges except those who are constructed based on $y$, i.e., Alice 
does not know the edges between $B$ and $C$. To take into account
the edges between $B$ and $C$, Alice has to communicate with Bob. To show a lower bound 
on the number of symbolic steps, we show next an upper bound on the number of bits
of communication between Alice and Bob needed 
to simulate a symbolic one-step operation on~$G$. With the simulation of the one-step
operations, the symbolic algorithm can be used as a communication protocol for
distinguishing whether $G$ has diameter $2$ or $3$ 
and thus by Lemma~\ref{lem:diasparse} to solve
the Set Disjointness problem. 
Whenever the symbolic algorithm 
performs a $\pre$ or $\post$ operation (which are equivalent on undirected graphs)
for a set $S$ that contains vertices of $B$ or $C$, then Alice can simulate this
one-step operation by specifying the vertices of $B$ and $C$ that are in $S$ 
with a bit vector of size $2 s$, where Bob answers with a bit vector, again 
of size $2 s$, that indicates all vertices that are adjacent to $(B \cup C) \cap S$.
Thus the communication protocol can simulate a symbolic algorithm that performs 
$T$ one-step operations with at most $4 s T$ bits of communication.
Hence we have by Theorem~\ref{th:lbdisj} that $4 s T \ge \Omega(k) = \Omega(s^2)$
and thus $T \ge \Omega(s) = \Omega(n)$. Together with Lemma~\ref{lem:diasparse},
this proves the following theorem.
Note that any $(3/2-\varepsilon)$-approximation 
algorithm for the diameter of a graph can distinguish between diameter 2 and 3.

\begin{theorem}\label{th:diasparse}
	Any \upbr{probabilistic bounded error or deterministic} symbolic algorithm that
	computes a $(3/2-\varepsilon)$-approximation 
	of the diameter of an undirected connected graph with $n$ vertices needs $\Omega(n)$
	symbolic one-step operations.
\end{theorem}

\section{Upper Bounds}\label{sec:upper}
In this work we present the following upper bounds.

\subsection{Upper Bounds for Computing Strongly Connected Components}
We revisit the symbolic algorithm of Gentilini et al.~\cite{GentiliniPP08} that computes 
the SCCs and present a refined analysis to show that it only requires 
$O(\sum_{C \in SCCs(G)} (\dia_C+1))$ symbolic operations, 
improving the previously known $O(\min(n, \dia \cdot |SCCs(G)|))$ bound 
and matching the lower bound of Theorem~\ref{thm:scc3} 
(details in Section~\ref{sec:sccalg}).

\begin{theorem}\label{th:sccalg}
The algorithm of Gentilini et al.~\cite{GentiliniPP08} computes the SCCs of a graph~$G$
with $O(|SCCs(G)| + \sum_{C \in SCCs(G)} \dia_C)$ symbolic operations.
\end{theorem}

\subsection{Upper Bounds for Liveness, Reachability, and Safety Objectives}\label{sec:verifalg}

The upper bounds for Reachability, Safety, Liveness, and co-Liveness, are
summarized in the following proposition, which is straightforward to obtain
as discussed below.

\begin{proposition}\label{lem:objectives_ub}
	Let $SCC$ be the number of symbolic steps required to compute the SCCs of the graph and let $\target$ be the set of target/safe vertices. Then
	$Reach(\target)$ can be solved with $O(\dia)$ symbolic operations; 
	$\Buchi{\target}$ can be solved with $O(\min(SCC+\dia,|\target|\cdot \dia))$ symbolic operations;
	$\coBuchi{\target}$ can be solved with $O(|\target| + \dia)$ symbolic operations; and 
	$\Safe{\target}$ can be solved with $O(|\target|)$ symbolic operations.
\end{proposition} 

\noindent\emph{Algorithm for Reachability.} 
Given a target set $\target$, we can easily compute the vertices that can reach $\target$ by iteratively applying $\pre$ operations
until a fixed-point is reached. By the definition of diameter, this requires only $O(\dia)$ symbolic operations.

\smallskip\noindent\emph{Algorithms for Liveness.}
A simple algorithm for Liveness first starts an algorithm for computing SCCs and whenever an SCC is reported it
tests whether the SCC contains one of the target vertices and if so adds all vertices of the SCC to the winning set.
Finally, after all SCCs have been processed, the algorithm computes all vertices that can reach the current winning set and adds
them to the winning set.
That is, in total the algorithm only needs $O(SCC+\dia)$ many symbolic operations where $SCC$ is the number of symbolic operations 
required by the $SCC$ algorithm.
An alternative algorithm for Liveness with $O(|\target|\cdot \dia)$ symbolic operations is as follows.
For each $v \in \target$ check with $O(\dia)$ symbolic one-step operations whether the vertex can reach itself and if not remove the vertex from $\target$.
Then do a standard reachability with the updated set $\target$ as target, again with $O(\dia)$ symbolic one-step operations.

\smallskip\noindent\emph{Algorithm for co-Liveness.}
Given a set $\target$ of safe vertices, an algorithm for co-Liveness first restricts the graph to the vertices of $\target$;
in the symbolic model this can be done by intersecting the outcome of each $\pre$ and $\post$ operation with $\target$.
One then uses an SCC algorithm and whenever a non-trivial SCC is reported, all its vertices are added to the 
winning set. Finally, after all SCCs have been processed, all vertices that can reach
the current winning set in the original graph are added to the winning set.
That is, in total the algorithm only needs $O(|\target| + \dia)$ many symbolic operations, 
where $|\target|$ comes from the linear number of symbolic operations 
required by the $SCC$ algorithm for the modified graph.

\smallskip\noindent\emph{Algorithm for Safety.}
Given a set $\target$ of safe vertices, an algorithm for safety first restricts the graph to the vertices of $\target$.
One then uses an SCC algorithm and whenever a non-trivial SCCs is reported, all its vertices are added to the 
winning set. Finally, after all SCCs have been processed, all vertices that can reach, within $\target$,
the current winning set are added to the winning set.
That is, in total the algorithm only needs $O(|\target|)$ symbolic operations 
as both the $SCC$ algorithm for the modified graph and reachability in the modified graph are in $O(\target)$. 
Also notice that reachability is not bounded by $O(\dia)$ as restricting the graph to vertices of $\target$ can change the diameter.

  Notice that none of the above algorithms stores all the SCCs, but processes one SCC at a time.
  That is, the algorithms themselves only need a constant number of sets plus the sets stored in the algorithm 
  for computing SCCs (which can be done with $O(\log n)$ many sets).

\subsection{Upper Bounds for Approximate Diameter}
We present a $(1+\varepsilon)$-approximation algorithm for the diameter
(for any constant $\varepsilon > 0$) that takes $\widetilde{O}(n \sqrt{\dia})$
symbolic operations (the $\widetilde{O}$-notation hides
logarithmic factors). 

\begin{theorem}\label{th:diaalg}
A $(1+\epsilon)$-approximation of the diameter of a directed graph for any constant 
$\epsilon > 0$ can be obtained
with $\widetilde{O}(n \sqrt{\dia})$ symbolic operations, using $O(1)$ many sets.
\end{theorem}

\smallskip\noindent\emph{Technical Overview $(1+\varepsilon)$-Approximation Algorithm.}
The symbolic algorithm is based on the $3/2$-approximation 
algorithm by Aingworth et al.~\cite{AingworthCIM99} for explicit graphs, we give a high-level
overview of the differences here. An expensive step in the algorithm of \cite{AingworthCIM99}
is the computation of $s$-partial BFS trees that contain $s$ vertices closest
to a root vertex~$v$ and can be determined with $O(s^2)$ explicit operations (this 
part was later replaced and improved upon by \cite{RodittyW13,ChechikLRSTW14}).
In the symbolic model computing $s$-partial BFS trees would be less efficient, however,
we can compute \emph{all} vertices with distance at most~$x$
from $v$ with only $O(x)$ many symbolic operations. The limitation of the approximation
ratio~$c$ to $3/2$ in the algorithm of \cite{AingworthCIM99} comes from having to deal with 
vertices for which less than $s$ vertices are within distance at most $\dia / (2c)$.
In the symbolic model we do not have to consider this case since with a budget 
of $O(x)$ operations we can always reach at least $x$ vertices (assuming for now that
the graph is strongly connected and $x < \dia$). Thus the algorithm simplifies to the
second part of their
algorithm, whose core part is to find a vertex within distance at most $x$ for 
each vertex of the graph by using a greedy approximation algorithm for dominating
set. However, in the symbolic model storing a linear number of sets is too costly,
hence we inherently use that we can recompute vertices at distance at most $x$
efficiently when needed. 
Details are presented in Section~\ref{sec:diaalg}.

\section{Discussion}\label{sec:discussion}

\subsection{SCCs and Verification Objectives} 

First, our results show that the symbolic SCC algorithm by Gentilini et al.~\cite{GentiliniPP08} is essentially optimal.
That is, we have the three upper bounds of $O(n)$, $O(|SCCs(G)\cdot \dia|)$ and $O(\sum_{C \in SCCs(G)} (\dia_C +1))$
and matching lower bounds of $\Omega(n)$ (Theorem~\ref{thm:scc2}), $\Omega(|SCCs(G)|\cdot \dia)$ (Theorem~\ref{thm:scc2}), 
and $\Omega(|SCC| + \sum_{C \in SCCs(G)} \dia_C)$ (Theorem~\ref{thm:scc3}).
\begin{table}[h]
  \caption{Results}\smallskip
  \label{tab:results_scc}
  \centering
  \begin{tabular}{l | c  c  c}
    \toprule
        &  \multicolumn{3}{c}{number of symbolic operations}\\
    \midrule
    SCC & $\Theta(n)$ &  $\Theta(|SCCs(G)| \cdot \dia)$ & $\Theta(\sum_{C \in SCCs(G)} (\dia_C +1))$\\
    \bottomrule
  \end{tabular}
\end{table}

Our results for the different kinds of verification objectives are summarized in Table~\ref{tab:results_objectives}. 
We have an interesting separation between the reachability objective and  the other objectives in terms of the diameter $\dia$.
While reachability can be solved with $O(\dia)$ symbolic operations, all the other
objectives require $\Omega(n)$ symbolic one-step operation on graphs of constant diameter.

When considering the diameters $\dia_C$ of the SCCs we get another separation.
There we have that Liveness and Reachability can be solved with $O(\sum_{C \in SCCs(G)} (\dia_C +1))$ many symbolic operations, while
Safety and co-Liveness requires $\Omega(n)$ symbolic one-step operations on strongly connected graphs with constant diameter.
This reflects the fact that in the standard algorithm for Safety and co-Liveness the SCC computation is performed on a modified graph.

\begin{table*}[t]
  \caption{Results}\smallskip
  \label{tab:results_objectives}
  \centering
  \begin{tabular}{l | c | c | c}
    \toprule
    Objective &  \multicolumn{3}{c}{number of symbolic operations in terms of}\\
	& $n$   &  $\dia$ \& $|T|$    & $\dia_C$   \\
    \midrule
    \Reach{\target} & $\Theta(n)$ & $\Theta(\dia)$ & $\Theta(\dia)$ \\
    \Safe{\target} & $\Theta(n)$ & $\Theta(|\target|)$ & $\omega(\sum_{C \in SCCs(G)} (\dia_C +1))$\\\
    \Buchi{\target} & $\Theta(n)$ &  $O(|\target|\cdot \dia)$ / $\Omega(|\target|+ \dia)$ & $\Theta(\sum_{C \in SCCs(G)} (\dia_C +1))$\\
    \coBuchi{\target} & $\Theta(n)$ &  $\Theta(|\target| + \dia)$ & $\omega(\sum_{C \in SCCs(G)} (\dia_C +1))$\\
    \bottomrule
  \end{tabular}
\end{table*}

\subsection{Approximate Diameter}\label{sec:discuss:dia}

For explicitly represented graphs a $3/2$-approximation of the diameter can
be computed in $\widetilde{O}(m \sqrt{n})$ time \cite{RodittyW13,ChechikLRSTW14},
while under the strong exponential time hypothesis no $O(n^{2 - o(1)})$
time algorithm exists to distinguish graphs of diameter~2 and~3 (i.e., no
$(3/2-\varepsilon)$-approximation can be obtained) \cite{RodittyW13}. 
The fastest exact algorithms take $\widetilde{O}(mn)$ time. 
While for explicitly represented graphs small, constant diameters are a hard case, 
the current results suggest that in the symbolic model the diameter
of graphs with constant diameter can be determined more efficiently than for 
large diameters, as both the upper bound for exact and approximate computation
of the diameter depend on the diameter of the graph and are linear when the
diameter is constant. While the threshold of an approximation ratio of $3/2$ 
appears in our (linear) lower bound, the current symbolic upper bounds do not show 
this behavior. Several interesting open questions remain: Is there a $o(n)$ 
c-approximation algorithm when $c \in [3/2, 2)$? Is there a linear 
$(1+\varepsilon)$-approximation algorithm for graphs with super-constant diameter?
Or are there better lower bounds?

\section{Detailed Proofs}

\subsection{Proofs of Section~\ref{sec:lb_SCCs}}\label{sec:lb_SCCs_proofs}

\begin{proof}[Proof of Lemma~\ref{lem:scc1}.]
 We have to show that $f(x,y)=1$ iff the graph constructed in Reduction~\ref{red:scc2} has exactly $\ell$ SCCs.\smallskip

 First notice that there are no edges from a set $V_j$ to a set $V_i$ if $i<j$ and thus 
 there are at least $\ell$ SCCs, independently of the actual values of $x$ and $y$.
  
 \noindent$\Rightarrow:$ If $f(x,y)=1$ then all possible edges exists and it is easy to verify that 
	        the SCCs of the graphs are exactly the sets $V_i$ for $0 \leq i < \ell$.
 
 \noindent$\Leftarrow:$ If $f(x,y)=0$ then there are $0 \leq j \leq \bar{k}$, $0 \leq i < \ell$ 
		such that there is no edge from $v_{i,j+1}$ to $v_{i,j}$.
		Now the set $V_i$ splits up in at least two SCCs and thus there are at least $\ell+1$ SCCs.
\end{proof}

\begin{proof}[Proof of Theorem~\ref{thm:scc3}.]
 We have to show that any \upbr{probabilistic bounded error or deterministic} symbolic algorithm
 that computes the SCCs needs $\Omega(|SCCs(G)| + \sum_{C \in SCCs(G)} \dia_C)$ symbolic one-step operations.\smallskip

 First consider Reduction~\ref{red:scc2} and notice that for the constructed graph $\sum_{C \in SCCs(G)} \dia_C \in \Theta(n)$.
 Then by Theorem~\ref{thm:scc2} we already have a $\Omega(\sum_{C \in SCCs(G)} (\dia_C))$ bound.
 Now consider Reduction~\ref{red:scc3} and notice that the constructed graph has $n$ SCCS iff $x$ and $y$ are disjoint.
 Now we can use the same argument as in the proof of Theorem~\ref{thm:scc2} that each
 symbolic one-step operations just needs constant communication.
 The instances where $x$ and $y$ are disjoint have $n$ SCCs and $\sum_{C \in SCCs(G)} \dia_C =0$.
 Hence an algorithm with $o(|SCCs(G)|)$ symbolic one-step operations
 would imply a communication protocol with $o(k)$ communication, a contradiction to Theorem~\ref{th:lbdisj}. 
 By combining the two lower bounds we get the desired $\Omega(|SCCs(G)| + \sum_{C \in SCCs(G)} \dia_C)$ bound.
\end{proof}

\subsection{Proofs of Section~\ref{sec:lb_objectives}}

\begin{proof}[Proof of Theorem~\ref{thm:reach}.]
  We have to show that any \upbr{probabilistic bounded error or deterministic} symbolic algorithm
  that solves Reachability in graphs with diameter $\dia$ requires $\Omega(\dia)$ symbolic one-step operations.\smallskip

  Consider the graph from Reduction~\ref{red:scc2} with parameter $\ell=1$.
  We have that $v_0$ is reachable from all vertices iff the graph is strongly connected.
  From the proof of Theorem~\ref{thm:scc2} we have that testing whether the graph 
  is strongly connected requires $\Omega(k)$ symbolic one-step operations.
  Now notice that 
  (a) the graph is strongly connected iff $v_k$ can reach $v_0$ and that 
  (b) if the graph is strongly connected then $\dia=k$.
\end{proof}

\begin{proof}[Proof of Theorem~\ref{thm:buchi}.]
 For (1) \& (3) consider the graph constructed in Reduction~\ref{red:buchi} and 
 the target set $\target=\{v_0, v_1, \dots, v_{k-1}\}$.
 We have a valid reduction from the Set Disjointness problem as
 the vertex $v_0$ is winning for $\Buchi{\target}$ iff there is a loop for one of the vertices in $\target$
 iff $S_x \cap S_y \not= \emptyset$.
 By the same argument as in the proof of Theorem~\ref{thm:scc2} we have that each
 symbolic one-step operations just needs constant communication.
 Hence an algorithm with $o(n)$, $o(|\target|)$ or $o(|SCCs(G)|)$ symbolic one-step operations
 would imply a communication protocol with $o(k)$ communication, a contradiction. 
 
 For (2) consider the graph constructed in Reduction~\ref{red:scc2} with $\ell=1$ and the target set $\target=\{v_0\}$.
 It is easy to verify that the vertex $v_k$ is winning if it can reach $\target$. 
 Thus the $\Omega(\dia)$ lower bound for reachability also applies here.
 Notice that this also gives a $\Omega(\sum_{C \in SCCs(G)} \dia_C)$ lower bound.
 
 The  $\Omega(|\target| + \dia)$ lower bound in (4) is a direct consequence of the $\Omega(\dia)$ lower bound for instances 
 with constant size target sets $\target$, and the $\Omega(|\target|)$ lower bound for instances with constant diameter~$\dia$.

 Finally, (5) is by the $\Omega(|SCCs(G)|)$ bound from  Reduction~\ref{red:buchi} and the $\Omega(\sum_{C \in SCCs(G)} \dia_C)$ bound
 by Reduction~\ref{red:scc2}  with $\ell=1$.
\end{proof}

\begin{proof}[Proof of Theorem~\ref{thm:coBuchi}.]
 1) \& 2) 
 Consider the graph constructed in Reduction~\ref{red:buchi} and 
 the set of safe vertices $\target=\{v_0, v_1, \dots, v_{k-1}\}$.
 We have a valid reduction from the Set Disjointness problem as
 the vertex $v_0$ is winning for $\Safe{\target}$ iff there is a loop for one of the vertices in $\target$
 iff $S_x \cap S_y \not= \emptyset$.
 By the same argument as in the proof of Theorem~\ref{thm:scc2} we have that each
 symbolic one-step operations just needs constant communication.
 Thus an algorithm with $o(n)$ or $o(|\target|)$ symbolic one-step operations
 would imply a communication protocol with $o(k)$ communication, a contradiction. 
 
 3) Consider the graph constructed in Reduction~\ref{red:buchi} but replace the edge $(v_k,v_k)$ by 
    the edge $(v_k,v_0)$.
    When considering the set of safe vertices $\target=\{v_0, v_1, \dots, v_{k-1}\}$ 
    the same arguments as above apply and thus 
    we get a $\Omega(n)$ lower bound.
    However, the graph is strongly connected and has diameter $2$ and 
    thus $|SCCs(G)| + \sum_{C \in SCCs(G)} \dia_C = O(1)$.
\end{proof}

\begin{proof}[Proof of Proposition~\ref{prop:coBuchi}.]
 Consider the graph constructed in Reduction~\ref{red:scc2} with $\ell=1$, add an additional loop edge $(v_0,v_0)$,
 and consider the set $\target=\{v_0\}$ of safe vertices, i.e., $v_0$ is the only safe vertex.
 It is easy to verify that the vertex $v_k$ is winning in $\coBuchi{\target}$ if it can reach $\target$. 
 Thus the $\Omega(\dia)$ lower bound for reachability also applies here.
\end{proof}

\subsection{Proofs of Section~\ref{sec:lb_diameter}}

\begin{proof}[Proof of Lemma~\ref{lem:diasparse}.]
	Let $G = (V, E)$ be the graph given by Reduction~\ref{red:diasparse} and let 
	$D$ denote its diameter. 
	We have to show that if $S_x \cap S_y = \emptyset$, then $D = 2$, otherwise $D = 3$.\smallskip

	First note that through the edges adjacent to the auxiliary vertices the diameter
	of $G$ is at most~$3$. Furthermore we have for all $0 \le i,j < s$ that
	$\dist(a_i, b_j) \le 2$, $\dist(u, c_j) = 2$, and $\dist(t, c_j) = 1$,
	for all $i \ne j$ additionally $\dist(a_i, a_j) = \dist(b_i, b_j) = 
	\dist(c_i, c_j) = 2$, and for all $v \in A \cup B$ it holds that
	$\dist(u, v) = 1$ and $\dist(t, v) = 2$. Thus whether
	$D$ is $2$ or $3$ depends only on the maximum over all $0 \le i,j < s$
	of $\dist(a_i, c_j)$ and $\dist(b_i, c_j)$. 
	
	If $S_x \cap S_y = \emptyset$, then for each pair of indices $0 \le i, j < s$ at least
	one of the edges $(a_i, c_j)$ and $(b_i, c_j)$ exists. Since $(a_i, b_i) \in E$,
	we have for all $0 \le j < s$ and all $v \in A \cup B$ that $\dist(v, c_j) \le 2$
	and hence $D = 2$.
	
	If $S_x \cap S_y \ne \emptyset$, let $\ell \in S_x \cap S_y$ and let $0 \le i, j < s$
	be such that $\ell = i \cdot s + j$. Then neither the edge $(a_i, c_j)$ nor the 
	edge $(b_i, c_j)$ exists. Thus the vertex $a_i$, and analogously the vertex $b_i$,
	has edges to the following vertices only: the auxiliary vertex~$u$,
	the vertex $b_i$, and vertices $c_{j'} \in C$ with $j' \ne j$. The vertex $c_j$
	has edges only to the auxiliary vertex~$t$ and to vertices $a_{i'} \in A$ and 
	$b_{i'} \in B$ with $i' \ne i$. Hence none of the vertices adjacent to $a_i$,
	and respectively for $b_i$, is adjacent to $c_j$ and thus we have $D = \dist(a_i, c_j)
	= \dist(b_i, c_j) = 3$.
\end{proof}

\subsection{An Improved Upper Bound for Strongly Connected Components}\label{sec:sccalg}
\noindent\emph{Result.}
Gentilini et al.~\cite{GentiliniPP08} provide a symbolic algorithm for computing the 
strongly connected components (SCCs) and show a bound of 
$O(\min(n, \dia \cdot |SCCs(G)|))$ on the number of its symbolic operations 
for a directed graph with $n$ vertices, diameter~$\dia$, and $|SCCs(G)|$ many SCCs.
Let $\dia_C$ be the diameter of an SCC~$C$. We give a tighter analysis of 
the algorithm of \cite{GentiliniPP08} that shows an upper bound of 
$O(\sum_{C \in SCCs(G)} (\dia_C+1))$
symbolic operations that matches our lower bound (Theorem~\ref{thm:scc3}).  
We have both 
$\sum_{C \in SCCs(G)} (\dia_C+1) \le (\dia + 1) \cdot |SCCs(G)|$
and $\sum_{C \in SCCs(G)} (\dia_C+1) \le n + |SCCs(G)| \le 2n$
and thus our upper bound is always at most the previous one.
We additionally observe that the algorithm can be implemented with $O(\log n)$
many sets (when the SCCs are output immediately and not stored). 
We first explain the intuition behind the algorithm of \cite{GentiliniPP08}
and then present the improved analysis of its number of symbolic steps.

\smallskip\noindent\emph{Symbolic Breadth-First Search.}
While explicit algorithms for SCCs are based on depth-first search (DFS), DFS is 
impractical in the symbolic model. However, breadth-first search (BFS) from 
a set $U \subseteq V$ can be 
performed efficiently symbolically, namely proportional to its depth, as defined below.
\begin{Definition}[Symbolic BFS]
	A \emph{forward search} from a set of vertices $U = U_0$ is given by a sequence 
	of $\post$ operations such that $U_i = U_{i-1} \cup \post(U_{i-1})$ for $i > 0$ until
	we have $U_i = U_{i-1}$. We call $U_i \setminus U_{i-1}$ the $i$-th 
	\emph{level} of the forward	search and the index of the last non-empty
	level the \emph{depth} $\depthout{U}$ of the forward search. Let 
	$\fw{U} = U_{\depthout{U}}$ be the \emph{forward set}, which is equal to the 
	vertices reachable from $U$. Analogously we define the \emph{backward search}
	of depth $\depthin{U}$ and the \emph{backward set} $\bw{U}$ for $\pre$ operations.
	We denote a singleton set $U = \set{u}$ by~$u$.
\end{Definition}
\noindent There is a simple algorithm for computing the SCCs 
symbolically with BFS that takes $O(\dia \cdot |SCCs(G)|)$ many symbolic steps: 
Start with an arbitrary vertex~$v$. Compute the SCC containing $v$ by taking 
the intersection of $\fw{v}$ and $\bw{v}$, remove the obtained SCC from the graph, 
and repeat.

\smallskip\noindent\emph{Skeleton-based Ordering.} 
The importance of DFS for SCCs lies in the \emph{order} in which the SCCs are computed. 
Starting from a vertex~$v$ that lies in an SCC without outgoing edges (i.e.\ a sink 
in the DAG of SCCs of the graph), the forward search does not leave the SCC and 
for computing SCCs the backward search can be restricted to the vertices of $\fw{v}$,
i.e., the SCC of $v$ can be determined with proportional to the diameter 
of the SCC many symbolic steps. The DFS-based SCC algorithm of \cite{Tarjan72}
finds such an SCC first. The algorithm of \cite{GentiliniPP08} is based on 
an ordering obtained via BFS that achieves a DFS-like ordering suitable for computing
SCCs symbolically. Our tighter analysis essentially shows that their approach achieves
the best ordering we can hope for. The ordering is given by so-called \emph{skeletons}.
\begin{Definition}
A pair $(\S,v)$ with $v \in V$ and $\S \subseteq V$ is a \emph{skeleton} of $\fw{u}$ for $u \in V$ if
$v$ has maximum distance from $u$ and the vertices of $\S$ form a shortest path from $u$ to $v$.
\end{Definition}
\noindent Let 
$SCC(v)$ denote the SCC containing~$v \in V$. 
The SCCs of the vertices of $\S$ 
will be computed in the following order: First $SCC(u)$ is computed by 
performing a backward search from~$u$ within $\fw{u}$. The remaining SCCs of 
of the vertices of $\S$ 
are then computed in the reverse order of 
the path induced by~$\S$, starting with $SCC(v)$. 
We now describe the overall algorithm, where in addition to $SCC(\S, v)$ the SCCs
of $V \setminus \fw{u}$ are computed (potentially using a different, previously
computed skeleton).

\smallskip\noindent\emph{The Algorithm.} The pseudo-code of the algorithm is given in Algorithm~\ref{alg:SymbolicSCC}, the pseudo-code for the sub-procedure for computing
a forward set including a skeleton in Algorithm~\ref{alg:Skeleton}.
Processing a graph $G$, the algorithm proceeds as follows: 
it starts from some vertex~$v$ and computes the set of 
reachable vertices, i.e, the forward set~$\fw{v}$, including a skeleton
that includes exactly one vertex of each level of the forward search and forms
a shortest path in~$G$.
It then starts a backward search starting from $v$ in the subgraph induced by~$\fw{v}$.
Clearly the SCC of $v$ is given by the vertices that are reached by the backward search. 
The algorithm returns this SCC as an SCC of~$G$ and recurses on 
(a) the subgraph $G_{V \setminus \fw{v}}$ induced by the 
vertices of $V \setminus \fw{v}$ and (b) the subgraph $G_{\fw{v} \setminus SCC(v)}$
induced by the vertices of $\fw{v} \setminus SCC(v)$. For the recursion on
(a) we update a potentially already existing skeleton by removing all vertices 
that are in the current SCC (initially we have an empty skeleton) 
while for the recursion on (b) we use the skeleton computed by the forward search (but also remove vertices of the current SCC).
The skeleton is then used to select the starting vertex in the consecutive steps of the
algorithm: when the algorithm is called with skeleton $(\S, v)$, then 
the forward search is started from~$v$; when the skeleton was computed in a forward
search from a vertex~$u$, then this corresponds to the vertex of the skeleton that 
is furthest away from~$u$ and contained in this recursive call.

\renewcommand{\S}{\mathcal{S}}

\begin{algorithm2e}[t]
	\SetAlgoRefName{SCC-Find}
	\caption{Symbolic SCC Algorithm}
	\label{alg:SymbolicSCC}
	\SetKwInOut{Input}{Input}
	\SetKwInOut{Global}{Global}
	\SetKwInOut{Output}{Output}
	\SetKw{break}{break}
	\SetKwFunction{Skel}{Skel\_Forward}
	\SetKwFunction{SCCFind}{SCC-Find}
	\SetKwProg{myproc}{Procedure}{}{}
	\BlankLine
	\Input{%
	  \emph{Graph} $G = (V, E)$,
	  \emph{Skeleton} $(\S,v)$
	}
	\BlankLine
	\If{$ V = \emptyset$} 
	    {
		return $\emptyset$; 
	    }
	\If{$ \S = \emptyset$} 
	    {
		$v \gets Pick(V)$ \tcp*{If there is no skeleton, pick arbitrary vertex}
	    }
	$(FW, \S',v') \gets \mbox{\Skel}(V,E,v)$  \tcp*{Forward search incl.\ skeleton} \medskip
	
	\tcc{Compute the SCC containing v}\smallskip
	$SCC \gets \{v\}$\;	
	\While{ $(\pre(SCC) \cap FW) \setminus SCC \not= \emptyset$ \label{alg:SymbolicSCC:while}}
	  {
	      $SCC \gets SCC \cup (\pre(SCC) \cap FW)$
	  }
	output $SCC$ as an SCC\;\medskip
	
	\tcc{Recursive calls}\smallskip
	$\mbox{\SCCFind}(G_{V \setminus FW},\ (\S \setminus SCC, (pre(SCC \cap \S) \setminus SCC) \cap \S))$\;
	
	$\mbox{\SCCFind}(G_{FW \setminus SCC},\ (\S' \setminus SCC, v'))$\;\smallskip
	
	\Return SCCs\; 
\end{algorithm2e}

\newcommand{\LEVEL}{\text{LEVEL}}
\begin{algorithm2e}[t]
	\SetAlgoRefName{Skel-Forward}
	\caption{Skeleton Forward Search Algorithm}
	\label{alg:Skeleton}
	\SetKwInOut{Input}{Input}
	\SetKwInOut{Output}{Output}
	\SetKw{break}{break}
	\SetKwFunction{Skel}{Skel\_Forward}
	\SetKwFunction{SCCFind}{SCC-Find}
	\SetKwProg{myproc}{Procedure}{}{}
	\BlankLine
	\Input{%
	  \emph{Graph} $G = (V, E)$,
	  \emph{Node} $v$
	}
	\Output
	{
	    $FW$ Set of vertices reachable from $v$;\\
	    $(\S', v')$ Skeleton for $FW$
	}
	\BlankLine
	$FW \gets \emptyset$; $i \gets 0$; $\LEVEL[0] \gets \{v\}$\;\medskip
	
	\tcc{Forward Search}\smallskip
	
	\While{ $\text{\emph{LEVEL}}[i] \not= \emptyset$ }
	  {
	      $FW \gets FW \cup \LEVEL[i]$\;
	      $i \gets i+1$\;
	      $\LEVEL[i] \gets \post(\LEVEL[i-1]) \setminus FW$\;
	  }\medskip
	
	\tcc{Compute Skeleton}\smallskip

	$i \gets i-1$\;
	$v' \gets Pick(\LEVEL[i])$\;
	$\S'\gets \{v'\}$\;
	\While{$i \ge 1$}
	  {
	      $i \gets i-1$\;
	      $\S'\gets \S' \cup \{Pick(\pre(\S') \cap \LEVEL[i])\}$\;
	      
	  }\smallskip
	  
	\Return $(FW, \S', v')$\; 
\end{algorithm2e}

\smallskip\noindent\emph{A Refined Analysis.}
The correctness of the algorithm is by \cite{GentiliniPP08}. 
Notice that the algorithm would be correct even without the usage of skeletons 
but the skeletons are necessary to make it efficient,
i.e., to avoid unnecessarily long forward searches.
We show the following theorem.
\begin{theorem}[Restatement of Theorem~\ref{th:sccalg}]
 With input $(G, \emptyset)$ Algorithm~\ref{alg:SymbolicSCC} computes 
 the SCCs of $G$ and requires $O(\sum_{C \in SCCs(G)} (\dia_C+1))$ symbolic operations.
\end{theorem}
\noindent 
The analysis of \cite{GentiliniPP08} of the number of symbolic steps of 
Algorithm~\ref{alg:SymbolicSCC}
uses that (1) each vertex is added to at most two skeletons, (2) the 
steps of the forward searches can be charged to the vertices in the skeletons,
and (3) backward searches are only performed to immediately identify an SCC and 
thus can be charged to the vertices of the SCC; hence
both the steps of the forward and of the backward searches can be bounded with $O(n)$.
For the backward searches it can easily be seen that the number of
$\pre$ operations to identify the SCC~$C$ is also bounded by 
$\dia_C + 1$. For the forward searches we show that 
each part of the skeleton (that in turn is charged for the forward search) can be 
charged to $\dia_C + 1$ for some SCC~$C$;
this in particular exploits that skeletons are shortest paths (in the graph
in which they are computed).
\begin{lemma}[\cite{GentiliniPP08}]\label{lem:shortestpath}
 For each recursive call of \ref{alg:SymbolicSCC} with input $G, (\S, v)$ we have 
 that $\S$ is a set of vertices that induces a shortest path in the graph~$G$ 
 and $v$ is the last vertex of this path.
\end{lemma}

We first recall the result from \cite{GentiliniPP08} that shows 
that the number of symbolic operations in an execution of \ref{alg:Skeleton}
is proportional to the size of the computed skeleton.
\begin{lemma}[\cite{GentiliniPP08}]\label{lem:skeleton}
 \ref{alg:Skeleton} only requires $O(\S')$ symbolic operations, i.e., is linear in the output,
  and can be implemented using only constantly many sets.
\end{lemma}
\begin{proof}
 For each level of the forward search we need one $\post$ operation in the first
 while loop and one $\pre$ operation in the second
 while loop, and the number of set operations is in the order of one-step operations.
 As for each level we add one vertex to $\S'$, the result follows.
 Moreover, there is no need to explicitly store all the levels as they can be 
 easily recomputed from the next level when needed, increasing the number of 
 symbolic operations only by a constant factor.
\end{proof}
\begin{Remark}
 Given Lemma~\ref{lem:skeleton} we can implement \ref{alg:SymbolicSCC}
 using $O(\log n)$ many sets at a time by recursing on the smaller of the
 two sub-graphs $G_{V \setminus FW}$ and $G_{FW \setminus SCC}$ first.
\end{Remark}

We split the cost of $O(\S')$ symbolic operations for \ref{alg:Skeleton} 
into two parts: the part $\S' \cap SCC(v)$ where $SCC(v)$ is the SCC 
identified in this level of recursion and the part $\S' \setminus SCC(v)$ that 
is passed to one of the recursive calls.
The following lemma shows that the first part and the subsequent backward search 
can be charged to $\dia_{SCC(v)} + 1$, using that $\S'$ is a shortest path in $\fw{v}$.

\begin{lemma}
 Without accounting for the recursive calls, each call of \ref{alg:SymbolicSCC}
 takes  $O(\dia_{SCC(v)} + 1 + |\S'\setminus SCC(v)|)$ symbolic operations, where
 $\S'$ is the new skeleton computed by \ref{alg:Skeleton}.
\end{lemma}
\begin{proof}
 By Lemma~\ref{lem:skeleton}, the call to \ref{alg:Skeleton} takes 
 $O(|\S'|)$ symbolic operations.
 In the $i$-th iteration of the while loop (Line~\ref{alg:SymbolicSCC:while}) 
 we add those vertices of $SCC(v)$ that can reach $v$ in $i$~steps. 
 That is, the loop terminates after $\dia_{SCC(v)}$ iterations and thus only requires
 $O(\dia_{SCC(v)}+1)$ many symbolic operations. 
 All the other steps just need a constant number of symbolic operations.
 That is, we have an upper bound of $O(\dia_{SCC(v)}+1+|\S'|)$.
 Now as $\S'$ induces a path starting at $v$ in $\fw{v}$,
 we have that whenever a vertex $u \ne v$ of $\S'$ can reach $v$, then also
 all vertices on the path from $v$ to $u$ can reach $v$
 and are therefore in the same SCC as~$v$. 
 Since the path is a shortest path, also every sub-path is a shortest path 
 and thus we have that $|\S' \cap SCC(v)| \leq \dia_{SCC(v)}+1$,
 i.e., $|\S'| \in O(\dia_{SCC(v)} + 1 + |\S'\setminus SCC(v)|)$.
 Hence we obtain the desired bound of $O(\dia_{SCC(v)} + 1 + |\S'\setminus SCC(v)|)$ 
 for the number of symbolic operations.
\end{proof}

Note that at each level of recursion we only charge the diameter of the SCC that
is output and the vertices of the newly computed skeleton. Thus we do not charge vertices of a 
skeleton again until they are contained in an SCC that is identified. The 
following lemma shows that in this case we can charge the symbolic steps 
that were charged to the vertices of the skeleton to $\dia_C + 1$,
where $C$ is the SCC the part of the skeleton belongs to.
Notice that $\S \setminus SCC(v)$ is the skeleton for the first recursive call and 
$\S'\setminus SCC(v)$ is the skeleton for the second recursive call, i.e., all vertices
of a skeleton are finally assigned to an SCC.
That is, we can bound the total number of symbolic steps by
$O(\sum_{C \in SCCs(G)} (\dia_C+1))$.
\begin{lemma}
 Whenever \ref{alg:SymbolicSCC} is called for a graph~$H$ and a skeleton $(\S,v)$,
 then $|\S \cap SCC(v)| \leq \dia_{SCC(v)}+1$ .
\end{lemma}
\begin{proof}
By Lemma~\ref{lem:shortestpath} the set $\S$ induces a shortest path in~$H$ that 
ends at~$v$. Thus if $v$ can reach a vertex $u \ne v$ of $\S$, then it can also 
reach all vertices of $\S$ that are on the path from $u$ to $v$ and all vertices
on this sub-path are in the same SCC as~$v$. Furthermore, the sub-path is a 
shortest path as well and thus the vertices $|\S \cap SCC(v)|$ form a 
shortest path in $SCC(v)$ and hence the diameter $\dia_{SCC(v)}$ of $SCC(v)$
 is at least $|\S \cap SCC(v)|-1$.
\end{proof}

\subsection{\texorpdfstring{$(1+\varepsilon)$}{(1+epsilon)}-Approximation of Diameter with \texorpdfstring{$\widetilde{O}(n \sqrt{D})$}{O(n * sqrt(D))} Symbolic Operations}
\label{sec:diaalg}

\smallskip\noindent\emph{Notation.}
Given a vertex $u \in V$, let $\distx{x}{u}$ denote the vertices with distance
at most $x$ \emph{from} $u$ and let $\distxto{y}{u}$ be the set of vertices with distance at most $y$ \emph{to} $u$. We have that $\set{u} \cup \post(\set{u}) = \distx{1}{u}$ and $\distx{x}{u} \cup \post(\distx{x}{u}) = \distx{x+1}{u}$. 
The maximum distance from the vertex $u$ to any other vertex $v \in V$ is 
given by the smallest~$x$ for which $\distx{x}{u} \cup 
\post(\distx{x}{u}) = \distx{x}{u}$.
Note that $x$ is at most $\dia$ and that computing~$x$ in this way corresponds
to performing a breadth-first-search (BFS) from $x$ on explicitly represented graphs;
thus following \cite{AingworthCIM99}, we denote the smallest $x$ for which
$\distx{x}{u} \cup \post(\distx{x}{u}) = \distx{x}{u}$ 
with $\depthout{u}$ and the smallest 
$y$ for which $\distxto{y}{u} \cup \pre(\distxto{y}{u}) = \distxto{y}{u}$
with $\depthin{u}$.
The set of vertices reachable from $u \in V$ is given by $\distx{\depthout{u}}{u}$.

\smallskip\noindent\emph{The Basic Exact Algorithm.}
The maximum  $\dist(u, v)$ over all pairs $u, v \in V$ for which $u$ can reach $v$ 
can be computed by taking the maximum of $\depthout{u}$ over all $u \in V$. 
Computing $\depthout{u}$ for all $u \in V$ takes $O(n \cdot \dia)$ many
$\post$ operations. To obtain only the value of $\dia$, only a constant number 
of sets have to be stored. Note that this basic algorithm only uses a linear number 
of symbolic operations for graphs with constant diameter. See Section~\ref{sec:lb_diameter} for
a matching lower bound for this case.

\smallskip\noindent\emph{A Simple 2-Approximation Algorithm.} 
If the graph $G$ is strongly connected, then a 2-approximation of 
$\dia$ is given by $(\depthout{u} + \depthin{u}) / 2$ for
any vertex $u \in V$. This follows from the triangle inequality and 
takes $O(\dia)$ many symbolic steps to compute.

\smallskip\noindent\emph{Result.}
We present an algorithm that computes an estimate $\adia$ of the diameter~$\dia$
of the input graph~$G$ such that $\adia \in [\dia-x ,\dia]$ for a parameter 
$x \le \sqrt{\dia}$ and takes $O(n \cdot \dia / x \log n )$ symbolic steps and uses a 
constant number of sets. For $x = \sqrt{\dia}$ this implies a bound of 
$O(n \sqrt{\dia} \log n)$ on the number of symbolic steps and an approximation guarantee
that is better than a $(1+\varepsilon)$-approximation for any constant 
$\varepsilon > 0$ (here we assume $\sqrt{\dia} \ge (1 + \varepsilon)/\varepsilon$; 
otherwise $\dia$ is constant anyway
and thus the exact algorithm only takes $O(n)$ symbolic steps). To pick 
the parameter $x$ correctly, one can use the 2-approximation algorithm
if the graph is strongly connected or use doubling search
at the cost of an additional factor of $\log n$.

\smallskip\noindent\emph{Searching from Neighborhood.} Let $a$ and $b$ be two vertices with maximum distance in $G$, i.e., $\dist(a,b) = \dia$.
We start with the simple observation that it is sufficient to determine the 
depth of a BFS from a vertex with distance at most $x$ from $a$ to obtain 
an estimate that is at most $x$ smaller than $\dia$.
\begin{observation}[see also \cite{AingworthCIM99}]
	Let $a,b \in V$ be such that $\dist(a,b) = \dia$. Then $\depthout{v} \ge \dia -x$
	for $v \in \distx{x}{a}$ and $\depthin{u} \ge \dia - y$ for $u \in \distxto{y}{b}$.
\end{observation}
\noindent Thus to obtain an estimate for the diameter, it is certainly sufficient 
to find a vertex~$u$ in $\distx{x}{v}$ for \emph{every} vertex $v\in V$ and compute 
$\depthout{u}$ for all these vertices.
If the graph is not strongly connected, it can happen that some vertices $v$ 
can not reach $x$ vertices and hence $\distx{x}{v}$ might contain less than $x$ 
vertices. In this case we know that $\depthout{v} < \dia$. Thus 
it also suffices to find a vertex~$u$ in $\distx{x}{v}$ for every vertex $v \in V$
for which $|\distx{x}{v}| \ge x$; we denote this set of vertices with $\highdeg{x}$.
\begin{corollary}\label{cor:approxdiagivenS}
	Let $S$ be a set of vertices such that $S \cap \distx{x}{v} \ne \emptyset$
	for all $v \in \highdeg{x}$ for $x < \dia$. Let $\adia = \max_{u \in S} \depthout{u}$. Then 
	$\adia \in [\dia - x, \dia]$. Given $S$, computing $\adia$ takes $O(|S| \cdot \dia)$ symbolic operations and storing $O(1)$ many sets.
\end{corollary}
\smallskip\noindent\emph{Dominate each Neighborhood.} 
An \emph{out-dominating set} for a set of vertices $A \subseteq V$ contains for each 
vertex of~$A$ either the vertex or one of its successors. 
Finding a set $S \subseteq V$ such that $S \cap \distx{x}{v} \ne \emptyset$
for all $v \in \highdeg{x}$ is equivalent to find an 
out-dominating set for all vertices with degree at least $x$ in the following graph: 
Let $\hat{G}$ be the graph obtained from $G$ by adding an edge from~$v$ to 
each vertex of $\distx{x}{v} \setminus \set{v}$ for all $v \in V$. 
In $\hat{G}$ every vertex of $\highdeg{x}$ has out-degree at least~$x$.
Thus an out-dominating set for $\highdeg{x}$ in $\hat{G}$ contains a 
vertex of $\distx{x}{v}$ for all $v \in \highdeg{x}$, i.e., for all vertices $v \in V$
with $|\distx{x}{v}| \ge x$. We adopt the classical greedy algorithm for 
dominating set to compute an out-dominating set in $\hat{G}$ with the following guarantees.
We prove Lemma~\ref{lem:domset} in the following subsection.
\begin{lemma}\label{lem:domset}
	An out-dominating set $S$ for the vertices of $\highdeg{x}$
	in $\hat{G}$ with $|S| \in O(n / x \cdot \log n)$ can 
	be found with $O(n \cdot x \cdot \log n)$ symbolic operations on $G$,
	storing $O(1)$ many sets.
\end{lemma}
\smallskip\noindent\emph{Overall Algorithm.} 
Hence our algorithm is as follows. First we find a set $S$ of size 
$O(n / x \cdot \log n)$ that contains a vertex of $\distx{x}{v}$ for every 
$v \in \highdeg{x}$ in $O(n \cdot x \cdot \log n)$
symbolic steps (Lemma~\ref{lem:domset}). Then we compute $\depthout{u}$ for 
all $u \in S$ with $O(n \cdot \dia / x \cdot \log n )$ many symbolic steps and 
return the maximum value of $\depthout{u}$ that was found
(Corollary~\ref{cor:approxdiagivenS}). Together with the observations 
at the beginning of this section we obtain the following theorem. 
The $\widetilde{O}$-notation hides the logarithmic factors.
\begin{theorem}[Restatement of Theorem~\ref{th:diaalg}]
	A $(1+\epsilon)$-approximation of the diameter of a directed graph for any constant 
	$\epsilon > 0$ can be obtained 
	with $\widetilde{O}(n \sqrt{\dia})$ symbolic operations, using $O(1)$ sets.
\end{theorem}

\subsubsection{Proof of Lemma~\ref{lem:domset}}

A \emph{fractional} out-dominating set of a set $A \subseteq V$ is a function
that assigns a weight $w_v \in [0,1]$ to each $v \in V$ such that for every 
$v \in A$ the sum of the weights over $v$ and its successors is at least one.
The size of a fractional out-dominating set is the sum of all weights~$w_v$.
For Lemma~\ref{lem:domset} we want to obtain an out-dominating set 
of $\highdeg{x}$. The vertices of $\highdeg{x}$ have out-degree at least $x$ in $\hat{G}$. 
Thus a fractional out-dominating set of $\highdeg{x}$ in $\hat{G}$
is obtained by assigning each vertex a weight of $1 / x$. The 
size of this fractional out-dominating set is $O(n/x)$. 
We show a greedy algorithm that finds an out-dominating set of $\highdeg{x}$
in $\hat{G}$ of size within a logarithmic factor of the optimal fractional 
out-dominating set, i.e., of size $O(n / x \cdot \log n)$.
The greedy algorithm is given in Algorithm~\ref{alg:dominatingSet} and is 
a modification of the greedy algorithm by \cite{Johnson74,Lovasz75,Chvatal79}
using an idea from 
\cite{BadanidiyuruV14}. We first describe the algorithm and show that it takes
$O(n \cdot x \cdot \log n)$ symbolic steps to output an out-dominating set
of $\highdeg{x}$ and then prove that the size of 
the obtained out-dominating set for $\highdeg{x}$ is within $O(\log n)$ of 
the optimal fractional solution.

\begin{algorithm2e}[t]
	\SetAlgoRefName{Dominating Set}
	\caption{Algorithm for Out-Dominating Set in $\hat{G}$}
	\label{alg:dominatingSet}
	\SetKwInOut{Input}{Input}
	\SetKwInOut{Output}{Output}
	\BlankLine
	\Input{%
	  {Graph} $G = (V, E)$, parameter $x$
	}
	\Output
	{
	    Set $S\subseteq V$ that contains a vertex of $\distx{x}{v}$ for all $v$ with $|\distx{x}{v}| \ge x$
	}
	\BlankLine
	$S \gets \emptyset$ \tcc*{dominating set}
	$C \gets \emptyset$ \tcc*{covered vertices}
	$j \gets \lfloor \log_2 n \rfloor$ \tcc*{size threshold}
	
	\For(\tcc*[f]{don't have to cover vertices that reach $< x$ vertices}){$v \in V$\label{l:excludelowdeg_start}}
	{
	  \If{$|\distx{x}{v}| < x$}{
		$C \gets C \cup \set{v}$\label{l:excludelowdeg_end}\;
	  }
	}
	
	\While{$j \geq 0$}
	  {
	      \For{$v \in V \setminus S$}
		{
		  \If{$|\distxto{x}{v} \setminus C| \geq 2^j$}
		    {
		      $S \gets S \cup \{v\}$\;
		      $C \gets C \cup \distxto{x}{v}$\;
		    }
		}
		$j \gets j-1$\;
	  }
	  \Return $S$\;
\end{algorithm2e}

Algorithm~\ref{alg:dominatingSet} takes the graph $G = (V, E)$ and the parameter~$x$ 
as input. Constructing the graph $\hat{G}$, i.e., storing the sets $\distx{x}{v}$ for 
all $v\in V$, is too costly. 
Note that there is a one-to-one correspondence 
between the edges of $\hat{G}$ and the paths of length $\le x$ in $G$ and that 
the union of a vertex~$v$ with its successors 
in $\hat{G}$ is given by $\distx{x}{v}$
and the union of a vertex~$v$ with its predecessors
in $\hat{G}$ is given by $\distxto{x}{v}$.
We recompute the sets $\distxto{x}{v}$ in
$O(x)$ symbolic steps per set when needed by the algorithm.

Observe that the set~$S$ is an out-dominating set of $\highdeg{x}$ in $\hat{G}$
if and only if $\cup_{v \in S} \distxto{x}{v} \supseteq \highdeg{x}$. We say that
the vertices of $\cup_{v \in S} \distxto{x}{v}$ are covered by the set~$S$.
In the algorithm the set~$S$ denotes the set of vertices added to the out-dominating 
set so far and the set~$C$ denotes the vertices that are covered by~$S$;
we additionally add the vertices of $V \setminus \highdeg{x}$ to $C$, as they 
do not have to be covered
(lines~\ref{l:excludelowdeg_start}--\ref{l:excludelowdeg_end}).

The main part of the algorithm consists of a while-loop with $O(\log n)$ many 
iterations. The variable~$j$ is initialized with $\lfloor \log_2 n \rfloor$
and is decreased after each iteration of the while-loop; in the last iteration of the 
while-loop we have $j = 0$.
In each iteration every vertex that is not yet in $S$ is considered
one after the other. 
For each vertex~$v \in V \setminus S$ the set
$\distxto{x}{v}$ is computed and the vertex is added to $S$ if the set $\distxto{x}{v}$ 
contains more than $2^j$ vertices that are not yet covered. 
When a vertex~$v$ is added to $S$, the vertices
of  $\distxto{x}{v}$ are marked as covered by adding them to $C$. In the 
last iteration we have $2^j = 1$ and thus all vertices of $\highdeg{x}$ 
that were not covered yet are added to $S$. Hence the returned set~$S$ is an 
out-dominating set of $\highdeg{x}$. In each iteration of the while-loop 
$O(n \cdot x)$ symbolic operations are used, thus the algorithm takes 
$O(n \cdot x \log n)$ symbolic steps in total, storing $O(1)$ many sets at a time.

It remains to show that the size of $S$ is within a factor of $O(\log n)$
of the size of an optimal fractional out-dominating set $S^*$ of $\highdeg{x}$ 
in $\hat{G}$ with weights $w_v \in [0,1]$ for all $v \in V$.
The outline of the proof is as follows: For each vertex~$v$ that the greedy 
algorithm adds to~$S$, we charge a total cost $\ge 1$ to the weights $w_u$ and 
show that the total cost charged to $w_u$ for each $u \in V$ is at most
$2 \cdot H_{ \Delta_u} \cdot w_u$, where $H_i \in O(\log i)$ is the $i$-th harmonic number and $\Delta_u$ is the in-degree of $u$ in $\hat{G}$. 
This implies that the number of vertices in~$S$ is bounded
by the sum of $ 2 \cdot w_u \cdot H_n$ over all $u \in V$, which proves the claim.

We charge the weights when adding $v$ to $S$ in 
Algorithm~\ref{alg:dominatingSet} as follows: For each vertex of
$\distxto{x}{v} \setminus C$ (i.e. the newly covered vertices) we consider all vertices~$u$ that contribute to 
the cover of this vertex and charge them $w_u / |\distxto{x}{v} \setminus C|$.
Note that all vertices of $V \setminus \highdeg{x}$ are contained in~$C$ and 
thus $\distxto{x}{v} \setminus C \subseteq V_x$, hence each vertex 
of $\distxto{x}{v} \setminus C$ is covered by the optimal 
fractional out-dominating set with a weight of at least one. Thus 
this charges at least $1 / |\distxto{x}{v} \setminus C|$
per vertex of $\distxto{x}{v} \setminus C$ and hence at least one per vertex
added to~$S$.
We have that for a fractional out-dominating set 
the set $\distxto{x}{u}$ is the set of vertices to whose covering the weight 
$w_u$ contributes.
Thus the charge for vertex~$u$ when adding~$v$ is given by
\begin{equation*}
	\frac{\left\lvert \left(\distxto{x}{v} \setminus C \right) \cap
	\distxto{x}{u}\right\rvert}
	{\left\lvert\distxto{x}{v}\setminus C\right\rvert} \cdot w_u\,.
\end{equation*} 

We finally show that each vertex $v$ is charged at most $w_v \cdot H_{n+1}$.
Let $j'$ be the value of $j$ when $v$ is added to $S$,
i.e., $\lvert \distxto{x}{v} \setminus C \rvert \ge 2^{j'}$. 
Note that if a vertex~$u$ is charged a non-zero
amount, then it is not contained in $C$ and therefore not in $S$. Hence 
we have that $u$ was not added to $S$ in the previous iteration of the 
while-loop and thus by the greedy condition 
$\lvert \distxto{x}{u} \setminus C \rvert \le 2^{j'+1}$.
Hence whenever $u$ is charged for a vertex $v$, we have
\begin{equation*}
	\left\lvert\distxto{x}{v}\setminus C\right\rvert 
	\ge \frac{1}{2} \left\lvert\distxto{x}{u}\setminus C\right\rvert\,,
\end{equation*}
and thus
\begin{equation*}
	\frac{\left\lvert \left(\distxto{x}{v} \setminus C \right) \cap \distxto{x}{u}\right\rvert}
	{\left\lvert\distxto{x}{v}\setminus C\right\rvert} \cdot w_u 
	\\ \le 
	\frac{2 \cdot \left\lvert \left(\distxto{x}{v} \setminus C \right) \cap \distxto{x}{u}\right\rvert}
	{\left\lvert\distxto{x}{u}\setminus C\right\rvert} \cdot w_u \,.
\end{equation*}
Now consider the vertex $u$ over the whole algorithm. The vertex~$u$ is charged 
for each vertex in $\distxto{x}{u}$ at most once. By the above it is charged at most 
$2 w_u/|\distxto{x}{u}|$ for the first vertex it is charged for, at most 
$2 w_u/(|\distxto{x}{u}|-1)$ for the second vertex, and at most 
$2 w_u/(|\distxto{x}{u}|-i + 1)$ for the $i$-th vertex.
Thus a vertex~$u$ with $|\distxto{x}{u}| = \Delta_u$ is charged at most $ 2 w_u \cdot H_{\Delta_u} \in O (w_u \log(\Delta_u) )$, and hence
we obtain an $O(\log(n))$ approximation of $S^*$.

\subparagraph*{Acknowledgements.}
All authors are partially supported by the Vienna
Science and Technology Fund (WWTF) through project ICT15-003.
K.~C.~is partially supported by the Austrian Science Fund (FWF)
NFN Grant No S11407-N23 (RiSE/SHiNE) and an ERC Start grant
(279307: Graph Games). V.~L.~is partially supported by the ISF grant \#1278/16
and an ERC Consolidator Grant (project MPM). 
For W.~D.~and M.~H.~the research leading to these results has received 
funding from the European Research Council under the European Union’s
Seventh Framework Programme (FP/2007-2013) / ERC Grant Agreement no. 340506.

\printbibliography[heading=bibintoc]

\end{document}